\documentclass[journal]{IEEEtran}
\usepackage{mathrsfs}

\usepackage{graphicx}
\usepackage{color}
\usepackage{colortbl}
\usepackage{psfrag}
\usepackage{subfigure}
\usepackage{amssymb}
\usepackage{epsfig}
\usepackage{pifont}
\usepackage{amsmath}
\usepackage{array}
\usepackage{multicol}
\usepackage{pifont}
\usepackage{indentfirst} 
\usepackage{amsfonts}
\usepackage{fancyhdr}
\usepackage{amscd}
\usepackage{bm}
\usepackage{extarrows}
\usepackage{multirow}

\newtheorem{theorem}{Theorem}
\newtheorem{proposition}{Proposition}

\newtheorem{corollary}{Corollary}
\newtheorem{lemma}{Lemma}

\makeatletter
\def\ifundefined{\@ifundefined}
\makeatother 

\makeatletter
\renewcommand{\citepunct}{,\penalty\@m\hskip.13emplus.1emminus.1em}
\renewcommand{\citedash}{\hbox{--}\penalty\@m}
\makeatother

\begin{document}
%
\title{On the Feasibility of Linear Interference Alignment for MIMO Interference Broadcast Channels with Constant Coefficients}

\author{Tingting~Liu,~\IEEEmembership{Member,~IEEE,}
        and~Chenyang~Yang,~\IEEEmembership{Senior Member, IEEE}
\thanks{Copyright (c) 2012 IEEE. Personal use of this material is permitted. However, permission to use this material for any other purposes must be obtained from the IEEE by sending a request to pubs-permissions@ieee.org.}
\thanks{The authors are with the School of Electronics and Information Engineering, Beihang University, Beijing 100191, China (e-mail: ttliu@ee.buaa.edu.cn; cyyang@buaa.edu.cn).}
\thanks{
Manuscript received July 1, 2012; revised October 28, 2012 and
February 7, 2013; accepted February 7, 2013. This work was supported
in part by the National Natural Science Foundation of China under Grant
61120106002 and Grant 61128002, by the International S$\&$T Cooperation
Program of China under Grant 2008DFA12100, and by the China Postdoctoral
Science Foundation under Grant 20110490007. The review of this paper was
coordinated by Dr. W. Utschick.}}

\markboth{IEEE TRANSACTIONS ON SIGNAL PROCESSING, ACCEPTED FOR PUBLICATION}%
{Liu \MakeLowercase{\textit{et al.}}: On the Feasibility of Linear Interference Alignment for MIMO Interference Broadcast Channels with Constant Coefficients}
\maketitle

\begin{abstract}
In this paper, we analyze the feasibility of linear interference
alignment (IA) for multi-input-multi-output (MIMO) interference
broadcast channel (MIMO-IBC) with constant coefficients. We pose and
prove the necessary conditions of linear IA feasibility for general
MIMO-IBC. Except for the proper condition, we find another necessary
condition to ensure a kind of \emph{irreducible interference} to be
eliminated. We then prove the necessary and sufficient conditions
for a special class of MIMO-IBC, where the numbers of antennas are
divisible by the number of data streams per user. Since finding an
invertible Jacobian matrix is crucial for the sufficiency proof, we
first analyze the impact of \emph{sparse structure} and
\emph{repeated structure} of the Jacobian matrix. Considering that
 for the MIMO-IBC the sub-matrices of the Jacobian matrix corresponding to the
transmit and receive matrices have different \emph{repeated
structure}, we find an invertible Jacobian matrix by constructing
the two sub-matrices separately. We show that for the MIMO-IBC where
each user has one desired data stream, a proper system is feasible.
For symmetric MIMO-IBC, we provide proper but infeasible region of
antenna configurations by analyzing the difference between the
necessary conditions and the sufficient conditions of linear IA
feasibility.
\end{abstract}


\IEEEpeerreviewmaketitle

\section{Introduction}
\IEEEPARstart{I}{nter-cell} interference (ICI) is a bottleneck for future cellular
networks to achieve high spectral efficiency, especially for
multi-input-multi-output (MIMO) systems. When multiple base stations
(BSs) share both the data and the channel state information (CSI),
network MIMO can improve the throughput remarkably
\cite{Foschini-NetworkMIMO-06}. When only CSI is shared, the ICI can
be avoided by the coordination among the BSs. In information
theoretic terminology, the scenario without the data sharing is a
MIMO interference broadcast channel (MIMO-IBC) when each BS
transmits to multiple users in its serving cell with same
time-frequency resource, and is a MIMO interference channel
(MIMO-IC) when each BS transmits to one user in its own cell.

To reveal the potential of the interference networks, significant
research efforts have been devoted to find the capacity region.
To solve such a challenging problem while capturing the essential nature
of the interference channel, various approaches have been proposed
to characterize the capacity approximately. Degrees of freedom (DoF) is the
first-order approximation of sum rate capacity at high
signal-to-noise ratio regime and also called as multiplexing gain,
which has received considerable attentions. When using the
break-through concept of interference alignment (IA)
\cite{Jafar_all_2011}, a $G$-cell MIMO-IC where each BS and each
user have $M$ antennas can achieve a DoF of $M/2$ per cell
\cite{Cadambe2008b}. For a two-cell MIMO-IBC where each cell has $K$
active users, each BS and each user have $M=K+1$ antennas, a per
cell DoF of $M$ can be achieved when $K$ approaches to infinity
\cite{Suh2011}. This result is surprising, because the DoF is the
same as the maximal DoF achievable by network MIMO but without data
sharing among the BSs. Encouraged by such a promising performance of
linear IA, many recent works strived to analyze the DoF for MIMO-IC
\cite{Jafar2007b,Jafar_IAchain_Kcell} and MIMO-IBC
\cite{Park2011,Kim2010,IBC_L_cell_Lovel} with various settings.

For the MIMO-IC or MIMO-IBC with constant coefficients (i.e.,
without symbol extension over time or frequency domain), to derive
the maximum DoF achieved by linear IA, it is crucial to analyze the
minimum numbers of transmit and receive antennas that guarantees the
IA to be feasible. Yet the feasibility analysis of linear IA for
general MIMO-IC and MIMO-IBC is still an open problem since the
problem was recognized in \cite{Gomadam2008}.

The feasibility analysis of linear IA feasibility includes finding
and proving the necessary and sufficient conditions. For the
necessary conditions, a \emph{proper condition} was first proposed
in \cite{Yetis2010} by relating the IA feasibility to the problem of
determining the solvability of a system represented by multivariate
polynomial equations. When the channels are \emph{generic} (i.e.,
drawn from a continuous probability distribution), the authors in
\cite{Luo2012,Tse2011} proved that the proper condition is one of
the necessary conditions for the IA feasibility of MIMO-IC. The
proper condition was then respectively provided for symmetric
MIMO-IBC\footnote{In general MIMO-IBC, each BS has $M_i$ antennas to
support $K_i$ users and each user has $N_{i_k}$ antennas to receive
$d_{i_k}$ data streams. In symmetric cases, $M_i=M$, $K_i=K$,
$N_{i_k}=N$ and $d_{i_k}=d$.} in \cite{IBC_Honig}, general MIMO-IBC
in \cite{IBC_L_cell_LJD} and partially-connected symmetric MIMO-IBC
in \cite{IBC_Guillaud}. Besides the proper condition, another class
of necessary conditions was found for MIMO-IC in
\cite{Jafar_IAchain_Kcell,Jafar_3cell,Tse_3cell_2011} and for
MIMO-IBC in \cite{IBC_Guillaud}.

To prove the sufficient conditions, two different approaches have
been employed. One is to find a closed-form solution for linear IA
\cite{Suh2011,Jafar_3cell,Tse_3cell_2011,Shin2011}, and the other is
to prove the existence of a linear IA solution
\cite{Luo2012,Tse2011,Win_IC,Gonzalez_IC}. Unfortunately, the
closed-form IA solutions are only available for finite cases, e.g.,
symmetric three-cell MIMO-IC \cite{Jafar_3cell,Tse_3cell_2011} and
symmetric two-cell MIMO-IBC with special antenna configurations
\cite{Suh2011,Shin2011}. To find the sufficient conditions for
general cases, various methodologies were employed in
\cite{Luo2012,Tse2011,Win_IC,Gonzalez_IC} to show when the IA
solutions exist. These studies all indicate that, the IA solution
exists when the mapping $f: \mathcal{V}\rightarrow\mathcal{H}$ is
\emph{surjective} \cite{Gonzalez_IC}\footnote{It is also called
\emph{full-dimensional} \cite{Luo2012}, \emph{dominant}
\cite{Tse2011} or \emph{algebraic independent} of the polynomials
\cite{Win_IC}.}, where $\mathcal{H}$ is the channel space and
$\mathcal{V}$ is the solution space. These studies also proved that
if $f$ is surjective for one channel realization $\pmb{H}_{0}$, it
will be surjective for \emph{generic} channels with probability one.
Furthermore, the authors in \cite{Luo2012} proved that if the
Jacobian matrix of $\pmb{H}_{0}$ is invertible, $f$ is surjective
for $\pmb{H}_{0}$. Along another line, the authors in
\cite{Tse2011,Win_IC} proved that if the first-order terms of the IA
polynomial equations with $\pmb{H}_{0}$ are linear independent, $f$
will be surjective for $\pmb{H}_{0}$. Interestingly, the matrix
composed of the first-order terms in \cite{Tse2011,Win_IC} happens
to be the Jacobian matrix in \cite{Luo2012}. Moreover, the matrix to
check the IA feasibility in \cite{Gonzalez_IC} is also a Jacobian
matrix, though in a form different from \cite{Luo2012}.
Consequently, all the analysis in
\cite{Luo2012,Tse2011,Win_IC,Gonzalez_IC} indicate that to prove the
existence of the IA solution for MIMO-IC, an invertible Jacobian
matrix needs to be found, either explicitly or implicitly. This
conclusion is also true for MIMO-IBC.

The way to construct an invertible Jacobian matrix depends on the
channel feature. So far, an invertible Jacobian matrix has only been
found for single beam MIMO-IC with general configurations and
multi-beam MIMO-IC in two special cases: 1) the numbers of transmit
and receive antennas are divisible by the number of data streams per
user \cite{Luo2012}, and 2) the numbers of transmit and receive
antennas are identical \cite{Tse2011}. For the general multi-beam MIMO-IC, and
for both single and multi-beam MIMO-IBC, the problem remains
unsolved, owing to their different channel features with the single
beam MIMO-IC.

The Jacobian matrix of MIMO IC and MIMO-IBC has two important
properties in structure: 1) \emph{sparse structure} (i.e., many
elements are zero) and 2) \emph{repeated structure} (i.e., some
nonzero elements are identical). In general, it is hard to construct
an invertible matrix with the \emph{sparse structure}
\cite{Jocobian_Graph}. For MIMO-IC, the study in \cite{Win_IC}
indicates that if the Jacobian matrix can be constructed as a
permutation matrix (which is invertible), the linear IA is feasible.
However, only for some cases, e.g., single beam MIMO-IC and the
special class of multi-beam MIMO-IC considered in \cite{Luo2012},
there exists a Jacobian matrix that can be set as a permutation
matrix. For other cases, due to the \emph{repeated structure}, the
Jacobian matrices cannot be set as permutation matrices. Until now
only the authors in \cite{Tse2011} constructed an invertible
Jacobian matrix for the MIMO-IC with $M=N$, but the construction
method cannot be extended to the cases beyond such a special case.
In fact, when the Jacobian matrix is not able to be set as a
permutation matrix, how to construct an invertible Jacobian matrix
is still unknown. This hinders the analysis for finding the minimal
antenna configuration to support the IA feasibility.

In this paper, we investigate the feasibility of linear IA for the
MIMO-IBC with constant coefficients. The main contributions are summarized as follows.
\begin{itemize}
  \item
The necessary conditions of the IA feasibility for a general
MIMO-IBC are provided and proved. Except for the proper condition,
we find another kind of necessary condition, which ensures a sort of
\emph{irreducible ICI} to be eliminated.\footnote{ For the ICIs
between a BS (or a user) and multiple users (or multiple BSs), if
the dimension of these ICIs cannot be reduced by designing the
receive matrices (or the transmit matrices), they are
\emph{irreducible ICIs}.} The existence conditions of the
\emph{irreducible ICIs} are provided.
  \item
The sufficient conditions of the IA feasibility for a special class
of MIMO-IBC are proved, where the numbers of transmit and receive
antennas are all divisible by the number of data streams of each
user. Although the considered setup is similar to the MIMO-IC in
\cite{Luo2012}, the channel features of the two setups differ. For
MIMO-IBC, the invertible Jacobian matrix cannot be set as a
permutation matrix due to the confliction with the \emph{repeated
structure}. Based on the observation that the sub-matrices of the
Jacobian matrix of MIMO-IBC corresponding to the transmit and
receive matrices have different \emph{repeated structures}, we
propose a general rule to construct an invertible Jacobian matrix
where the two sub-matrices are constructed in different ways.
  \item
  From the insight provided by analyzing the necessary conditions
and the sufficient conditions, we provide the proper but infeasible region of antenna
configuration for symmetric MIMO-IBC.
\end{itemize}

The rest of the paper is organized as follows. We describe the
system model in Section \ref{Sec:System model}. The necessary
conditions for general MIMO-IBC and the necessary and sufficient
conditions for a special class of MIMO-IBC will be provided and
proved in Section \ref{Sec:Necessary_Condition} and Section
\ref{Sec:Sufficient_Condition}, respectively. We discuss the
connection between the proper condition and the feasibility
condition in Section \ref{Sec:Discussion}. Conclusions are given in
the last section.

\emph{Notations:} Conjugation, transpose, Hermitian
transpose, and expectation are represented by $(\cdot)^{*}$,
$(\cdot)^{T}$, $(\cdot)^{H}$, and $\mathbb{E}\{\cdot\}$,
respectively. $\mathrm{Tr}\{\cdot\}$ is the trace of a matrix, and $\mathrm{diag}\{\cdot\}$ is a block diagonal matrix. $\otimes$ is the Kronecker product operator, $\mathrm{vec}\{\cdot\}$ is the operator that converts a matrix or set into a column vector, $\pmb{I}_{d}$ is an identity matrix of size
$d$. $|\cdot|$ is the cardinality of a set, $\varnothing$ denotes an empty set, and $\mathcal{A}\setminus\mathcal{B}=\{x\in\mathcal{A}|  x\notin\mathcal{B}\}$ denotes the relative complement of $\mathcal{A}$ in $B$. $\exists$ means ``there exists'' and $\forall$ means ``for all''.

\section{System model}\label{Sec:System model}
Consider a downlink $G$-cell MIMO network. In cell $i$, BS$_i$
supports $K_i$ users, $i = 1, \cdots, G$. The $k$th user in cell $i$
(denoted by MS$_{i_k}$) is equipped with $N_{i_k}$ antennas to
receive $d_{i_k}$ desired data streams from BS$_i$,
$k=1,\cdots,K_i$. BS$_i$ is equipped with $M_i$ antennas to transmit
overall $d_i=\sum_{k=1}^{K_i}d_{i_k}$ data streams. The total DoF to
be supported by the network is $d^{\mathrm{tot}}=
\sum_{i=1}^{G}d_i=\sum_{i=1}^{G}\sum_{k=1}^{K_i}d_{i_k}$. Assume
that there are no data sharing among the BSs and every BS has
perfect CSIs of all links. This is a scenario of general MIMO-IBC,
and the configuration is denoted by $\prod_{i=1}^{G}(M_i\times
\prod_{k=1}^{K_i}(N_{i_k},d_{i_k}))$.

The desired signal of MS$_{i_k}$ can be estimated as
\begin{align}\label{Eq:Rx signal model}
\hat{\pmb{x}}_{i_k} = &\pmb{U}_{i_k}^{H}\pmb{H}_{i_k,i}\pmb{V}_{i_k}\pmb{x}_{i_k}+
\sum_{l=1,l\neq
k}^{K_i}\pmb{U}_{i_k}^{H}\pmb{H}_{i_k,i}\pmb{V}_{i_l}\pmb{x}_{i_l}\nonumber\\
&+\sum_{j=1,j\neq
i}^{G}\pmb{U}_{i_k}^{H}\pmb{H}_{i_k,j}\pmb{V}_j\pmb{x}_j+\pmb{U}_{i_k}^{H}\pmb{n}_{i_k}
\end{align}
where
$\pmb{x}_{i_k}\in \mathbb{C}^{d_{i_k}\times 1}$ is the symbol vector for MS$_{i_k}$ satisfying $\mathbb{E}\{\pmb{x}_{i_k}^{H}\pmb{x}_{i_k}\}=Pd_{i_k}$, $P$ is the
transmit power per symbol, and $\pmb{x}_j =
[\pmb{x}_{j_1}^{T},\cdots,\pmb{x}_{j_{K_i}}^{T}]^T $ is
the symbol vector for the $K_j$ users in cell $j$,
${\pmb{V}}_{i_k}\in \mathbb{C}^{M_i\times d_{i_k}}$ is the transmit matrix for MS$_{i_k}$ satisfying
  $\mathrm{Tr}\{{\pmb{V}}_{i_k}^{H}{\pmb{V}}_{i_k}\}=d_{i_k}$, and ${\pmb{V}}_j=[ {{{\pmb{V}}_{j_1}}, \cdots
,{\pmb{V}}_{j_{K_i}}} ]$ is the transmit matrix of BS$_j$ for the
$K_j$ users in cell $j$, $\pmb{U}_{i_k}\in \mathbb{C}^{N_{i_k}\times
d_{i_k}}$ is the receive matrix for MS$_{i_k}$,
${\pmb{H}}_{i_k,j}\in \mathbb{C}^{N_{i_k}\times M_j}$ is the channel
matrix of the link from BS$_j$ to MS$_{i_k}$ whose elements are
independent random variables with a continuous distribution, and
${\pmb{n}}_{i_k}\in \mathbb{C}^{N_{i_k}\times 1}$ is an additive
white Gaussian noise.

The received signal of each user contains the multiuser interference
(MUI) from its desired BS and the ICI from its interfering BSs,
which are the second and third terms in \eqref{Eq:Rx signal model}.
Without symbol extension, the \emph{linear IA conditions}
\cite{Yetis2010} can be obtained from \eqref{Eq:Rx signal model} as
follows,
\begin{subequations}
\begin{align}
\label{Eq:Constraint_IDI_free0}
\mathrm{rank} \left(\pmb{U}_{i_k}^{H}\pmb{H}_{i_k,i}\pmb{V}_{i_k}\right)&=d_{i_k},~\forall i,k\\
\label{Eq:Constraint_MUI_free0}
\pmb{U}_{i_k}^{H}\pmb{H}_{i_k,i}\pmb{V}_{i_l}&=\pmb{0},~\forall k\neq l\\
\label{Eq:Constraint_ICI_free0}
\pmb{U}_{i_k}^{H}\pmb{H}_{i_k,j}\pmb{V}_j&=\pmb{0},~\forall i\neq j
\end{align}
\end{subequations}
The polynomial equation \eqref{Eq:Constraint_IDI_free0} is a rank
constraint to convey the desired signals for each user. It can be
interpreted as a constraint in single user MIMO system: the
inter-data stream interference (IDI)-free transmission constraint.
\eqref{Eq:Constraint_MUI_free0} and \eqref{Eq:Constraint_ICI_free0}
are the zero-forcing (ZF) constraints to eliminate the MUI and ICI,
respectively.

Note that multiple data streams transmitted from one BS to a user
undergo the same channel. This leads to two features of MIMO-IBC, according
to the cases where the BS and the user are located in the same cell
or in different cells, which are
\begin{itemize}
  \item
\emph{Feature 1}: the desired signal and the MUI experienced at each
user undergoing the same channel.
  \item
\emph{Feature 2}: the multiple ICIs generated from one BS to a user
in other cell undergo the same channel even when each user only
receives one desired data stream.\footnote{This feature does not
appear in single beam MIMO-IC. By contrast, the feature appears in
both single beam and multi-beam MIMO-IBC.}
\end{itemize}

\section{Necessary Conditions for General Cases}\label{Sec:Necessary_Condition}
In this section, we present and prove the necessary conditions of
linear IA feasibility for general MIMO-IBC. Since the IA conditions
in (\ref{Eq:Constraint_IDI_free0})-(\ref{Eq:Constraint_ICI_free0})
are similar to MIMO-IC and the proof builds upon the same line of
the work in \cite{Luo2012}, we emphasize the difference of MIMO-IBC
from MIMO-IC, which comes from the first feature of MIMO-IBC.

\begin{theorem}[Necessary Conditions]\label{Theorem:Necessary_Condition}
For a general MIMO-IBC with configuration $\prod_{i=1}^{G}(M_i\times \prod_{k=1}^{K_i}(N_{i_k},d_{i_k}))$
where the channel matrices $\{\pmb{H}_{i_k,j}\}$ are generic
(i.e., drawn from a continuous probability distribution), if the linear IA is feasible, the following conditions must be
satisfied,
\begin{subequations}
\begin{align}
\label{Eq:Necessary_Condition_Signal}
&\min\{M_i-d_i,N_{i_k}-d_{i_k}\}\geq 0,~\forall i,k \\
\label{Eq:Necessary_Condition_Interference}
&\sum_{j:(i,j)\in \mathcal{I}}
\left(M_j-d_j\right)d_j+ \sum_{i:(i,j)\in
\mathcal{I}}\sum_{k\in
\mathcal{K}_i}\left(N_{i_k}-d_{i_k}\right)d_{i_k}\nonumber \\
\geq& \sum_{(i,j)\in\mathcal{I}}d_j\sum_{k\in
\mathcal{K}_i}d_{i_k},~\forall \mathcal{I}\subseteq\mathcal{J}\\
\label{Eq:Necessary_Condition_Compress}& \max\Big\{\sum_{j\in\mathcal{I}_{\mathrm{A}}}M_j,\sum_{i\in
\mathcal{I}_{\mathrm{B}}}\sum_{k\in
\mathcal{K}_i}N_{i_k}\Big\}\nonumber \\
\geq&
\sum_{j\in\mathcal{I}_{\mathrm{A}}}d_j+
\sum_{i\in\mathcal{I}_{\mathrm{B}}}\sum_{k\in
\mathcal{K}_i}d_{i_k},~\forall \mathcal{I}_{\mathrm{A}}\cap\mathcal{I}_{\mathrm{B}}=\varnothing
\end{align}
\end{subequations}
where $\mathcal{K}_i\subseteq \{1,\cdots,K_i\}$ is an arbitrary
subset of the users in cell $i$, $\mathcal{J}=\{(i,j)|1\leq i\neq j
\leq G\}$ denotes the set of all cell-pairs that mutually
interfering each other, $\mathcal{I}$ is an arbitrary subset of
$\mathcal{J}$, and
$\mathcal{I}_{\mathrm{A}},~\mathcal{I}_{\mathrm{B}}\subseteq
\{1,\cdots,G\}$ are arbitrary two subsets of the index set of
the $G$ cells.
\end{theorem}

\subsection{Proof of \eqref{Eq:Necessary_Condition_Signal}}
\begin{proof}
Comparing \eqref{Eq:Constraint_IDI_free0} and
\eqref{Eq:Constraint_MUI_free0}, we can see that the
channel matrices of MS$_{i_k}$ in the two equations are all equal to
$\pmb{H}_{i_k,i}$. As a result, the rank constraint is coupled with
the MUI-free constraint, such that the proof for MIMO-IC in
\cite{Gomadam2008,Yetis2010} cannot be directly applied.

Note that from the view of MS$_{i_k}$, the desired data streams of
other users in cell $i$ are its MUI, while from the view of BS$_i$,
all the data streams for the users in cell $i$ are its desired
signals. Combining \eqref{Eq:Constraint_IDI_free0} and
\eqref{Eq:Constraint_MUI_free0}, we can obtain a rank constraint for
BS$_i$ as $\mathrm{rank}
\left([\pmb{H}_{i_{1},i}^{H}\pmb{U}_{i_1},\cdots,
\pmb{H}_{i_{K_i},i}^{H}\pmb{U}_{i_{K_i}}]^{H}\pmb{V}_i\right)
=\sum_{k=1}^{K_i}d_{i_k}=d_i$. Then, the IA conditions in
\eqref{Eq:Constraint_IDI_free0}$-$\eqref{Eq:Constraint_ICI_free0}
can be equivalently rewritten as
\begin{subequations}
\begin{align}
\label{Eq:Constraint_MUI_free}
\mathrm{rank} \left(
\left[\begin{array}{ccc}
        \pmb{U}_{i_1}^{H} & {} & {\pmb{0}} \\
        {} & \ddots & {} \\
        {\pmb{0}} & {} & \pmb{U}_{i_{K_i}}^{H}
      \end{array}
\right]
\left[\begin{array}{c}
        \pmb{H}_{i_{1},i} \\
        \vdots \\
        \pmb{H}_{i_{K_i},i}
      \end{array}\right]
\pmb{V}_i\right)&=d_i,~\forall i\\
\label{Eq:Constraint_ICI_free}
\pmb{U}_{i_k}^{H}\pmb{H}_{i_k,j}\pmb{V}_j&=\pmb{0},~\forall i\neq j
\end{align}
\end{subequations}

Now the channel matrices in \eqref{Eq:Constraint_MUI_free} are
independent of those in \eqref{Eq:Constraint_ICI_free}.
Since the channel matrix $\pmb{H}_{i_k,i}$ is generic,
\eqref{Eq:Constraint_MUI_free} is automatically satisfied with
probability one when $\mathrm{rank}(\pmb{V}_i)=d_i$ and
$\mathrm{rank}(\pmb{U}_{i_k})= d_{i_k}$\cite{Gomadam2008}.
Therefore, \eqref{Eq:Necessary_Condition_Signal} is necessary
to satisfy the equivalent rank constraint
\eqref{Eq:Constraint_MUI_free}.
\end{proof}
Different from multi-beam MIMO-IC, the aggregated receive matrix for different users in MIMO-IBC has a block-diagonal structure due to the non-cooperation among the users.

The intuitive meaning of \eqref{Eq:Necessary_Condition_Signal} is straightforward. BS$_i$ should have
enough antennas to transmit the overall $d_i$ desired signals to
multiple users in cell $i$, i.e., to ensure MUI-free transmission,
and MS$_{i_k}$ should have enough antennas to receive its $d_{i_k}$
desired signals, i.e., to ensure IDI-free transmission.

\subsection{Proof of \eqref{Eq:Necessary_Condition_Interference}}
\begin{proof}
To satisfy \eqref{Eq:Constraint_ICI_free} under the constraint of
\eqref{Eq:Constraint_MUI_free}, we need to first reserve some
variables in the transmit and receive matrices to ensure the
equivalent rank constraints, and then use the remaining variables to
remove the ICI. To this end, we partition the transmit and receive
matrices as follows
\begin{align}\label{effective txrx}
  \pmb{V}_j & = \pmb{P}_j^{V}\left[\begin{array}{c}
   \pmb{I}_{d_j}\\
   \bar{\pmb{V}}_j
 \end{array}
  \right]\pmb{Q}_j^{V}, \quad
  \pmb{U}_{i_k} = \pmb{P}_{i_k}^{U}\left[\begin{array}{c}
       \pmb{I}_{d_{i_k}} \\
       \bar{\pmb{U}}_{i_k}
     \end{array}
  \right]\pmb{Q}_{i_k}^{U}
\end{align}
where $\pmb{P}_j^{V}\in \mathbb{C}^{M_j\times M_j}$ and
$\pmb{P}_{i_k}^{U}\in \mathbb{C}^{N_{i_k}\times N_{i_k}}$ are square
permutation matrices, $\pmb{Q}_j^{V}\in \mathbb{C}^{d_j\times d_j}$
and $\pmb{Q}_{i_k}^{U}\in \mathbb{C}^{d_{i_k}\times d_{i_k}}$ are
invertible matrices, and $\bar{\pmb{V}}_i\in
\mathbb{C}^{(M_j-d_j)\times d_j}$ and $\bar{\pmb{U}}_{i_k}\in
\mathbb{C}^{(N_{i_k}-d_{i_k})\times d_{i_k}}$ are the \emph{effective
transmit and receive matrices}, whose elements are the remaining
variables after extracting $d_j^2$ and $d_{i_k}^2$ variables of
${\pmb{V}}_j$ and ${\pmb{U}}_{i_k}$,
respectively.

Then, \eqref{Eq:Constraint_ICI_free} can be rewritten as
\begin{align}\label{Eq:Constraint_ICI_free_New0}
\left(\pmb{Q}_{i_k}^{U}\right)^{H}\left[\begin{array}{cc}
       \pmb{I}_{d_{i_k}} &  \bar{\pmb{U}}_{i_k}^{H}
     \end{array}
  \right]
\underbrace{\left(\pmb{P}_{i_k}^{U}\right)^{H}\pmb{H}_{i_k,j}
\pmb{P}_j^{V}}_{\bar{\pmb{H}}_{i_k,j}}\left[\begin{array}{c}
   \pmb{I}_{d_j}\\
   \bar{\pmb{V}}_j
 \end{array}
  \right]\pmb{Q}_j^{V}=\pmb{0}
\end{align}
where $\bar{\pmb{H}}_{i_k,j}=\left(\pmb{P}_{i_k}^{U}\right)^{H}\pmb{H}_{i_k,j}
\pmb{P}_j^{V}$ is the \emph{effective channel matrix}.

Further partition the effective channel matrix as follows
\begin{align*}
\bar{\pmb{H}}_{i_k,j}= \left[\begin{array}{cc}
\bar{\pmb{H}}_{i_k,j}^{(1)}& \bar{\pmb{H}}_{i_k,j}^{(2)}\\
\bar{\pmb{H}}_{i_k,j}^{(3)}& \bar{\pmb{H}}_{i_k,j}^{(4)}
\end{array}\right]
\end{align*}
where $\bar{\pmb{H}}_{i_k,j}^{(1)}\in \mathbb{C}^{d_{i_k}\times
d_j}$, $\bar{\pmb{H}}_{i_k,j}^{(2)}\in \mathbb{C}^{d_{i_k}\times
(M_j-d_j)}$, $\bar{\pmb{H}}_{i_k,j}^{(3)}\in
\mathbb{C}^{(N_{i_k}-d_{i_k})\times d_j}$ and
$\bar{\pmb{H}}_{i_k,j}^{(4)}\in \mathbb{C}^{(N_{i_k}-d_{i_k})
\times(M_j-d_j)}$, respectively.

Then,
\eqref{Eq:Constraint_ICI_free_New0} is equivalent to the following
equation,
\begin{align}\label{Eq:Constraint_ICI_free_New1}
\left[\begin{array}{cc}
          \pmb{I}_{d_{i_k}} & \bar{\pmb{U}}_{i_k}^{H}
        \end{array}
   \right]
\left[\begin{array}{cc}
\bar{\pmb{H}}_{i_k,j}^{(1)}& \bar{\pmb{H}}_{i_k,j}^{(2)}\\
\bar{\pmb{H}}_{i_k,j}^{(3)}& \bar{\pmb{H}}_{i_k,j}^{(4)}
\end{array}\right]
\left[\begin{array}{c}
   \pmb{I}_{d_j}\\
   \bar{\pmb{V}}_j
 \end{array}
  \right]=\pmb{0}
\end{align}
Now the IA conditions in \eqref{Eq:Constraint_MUI_free} and
\eqref{Eq:Constraint_ICI_free} turns into a single condition.

From (\ref{Eq:Constraint_ICI_free_New1}), the relationship between
the effective transmit and receive matrices and the effective channel matrices can be
expressed in the form of implicit function, i.e.,
\begin{align}\label{Eq:Constraint_ICI_Transmission_free}
\pmb{F}_{i_k,j}(\bar{\pmb{H}};\bar{\pmb{V}},\bar{\pmb{U}})
=&\bar{\pmb{H}}_{i_k,j}^{(1)}+\bar{\pmb{H}}_{i_k,j}^{(2)}\bar{\pmb{V}}_j+ \bar{\pmb{U}}_{i_k}^{H}\bar{\pmb{H}}_{i_k,j}^{(3)}
+\bar{\pmb{U}}_{i_k}^{H}\bar{\pmb{H}}_{i_k,j}^{(4)}\bar{\pmb{V}}_j\nonumber \\
=&\pmb{0},~\forall i\neq j
\end{align}
where $\pmb{F}_{i_k,j}(\cdot)$
represents the ICIs from $\bar{\pmb{V}}_j$ to $\bar{\pmb{U}}_{i_k}$,
i.e., the interference generated by the effective transmit matrix of
BS$_j$ to the $k$th user in cell $i$.

In \eqref{Eq:Constraint_ICI_Transmission_free},
$\pmb{F}_{i_k,j}(\cdot) \in \mathbb{C}^{d_{i_k}\times d_j}$ includes
$d_jd_{i_k}$ ICIs, and $\bar{\pmb{V}}_j$ and $\bar{\pmb{U}}_{i_k}$ provide $(M_j-d_j)d_j$
and $(N_{i_k}-d_{i_k})d_{i_k}$ variables, respectively. Hence,
\eqref{Eq:Necessary_Condition_Interference} ensures that for all
subsets of the equations in
\eqref{Eq:Constraint_ICI_Transmission_free}, the number of
involved variables is at least as large as the number of
corresponding equations. Analogous
to MIMO-IC, to eliminate the ICI in the network thoroughly, all the cell-pairs that
are interfering each other, i.e., those in set
$\mathcal{J}$ and any subset of it, $\mathcal{I}$, should be
considered. Different from MIMO-IC, we should not generate ICI to arbitrary subsets of the users
in each cell, i.e., $\mathcal{K}_i$, rather than not generate ICI to a single user in each cell.\footnote{For example, when $G=3$,
$\mathcal{J}=\{(1,2), (1,3), (2,1), (2,3), (3,1)$, $(3,2)\}$. If
$\mathcal{I}=\{(1,2)\}$, the right-hand side of
\eqref{Eq:Necessary_Condition_Interference} is $d_2\sum_{k\in
\mathcal{K}_{1}}d_{1_k}$, which is the number of ICIs from BS$_2$ to
the users in $\mathcal{K}_{1}$ (an arbitrary subset of the users in
cell 1), and the left-hand side is $(M_2-d_2)d_2+\sum_{k\in
\mathcal{K}_{1}}(N_{1_k}-d_{1_k})d_{1_k}$, which is the number of
variables to eliminate these ICIs.}

From the definition in \cite{Yetis2010}, we know that this is
actually the condition to ensure the MIMO-IBC to be proper. In a
MIMO-IC with generic channel matrix, the proper condition has been
proved as a necessary condition of the IA feasibility
\cite{Luo2012}. For the considered MIMO-IBC, the channel matrix
$\pmb{H}_{i_k,j}$ is also generic. Therefore, the proper condition
is necessary for the MIMO-IBC to be feasible. Now,
\eqref{Eq:Necessary_Condition_Interference} is proved.
\end{proof}

The intuitive meaning of \eqref{Eq:Necessary_Condition_Interference}
is to ensure that any pairs of BSs and the users in any pairs of
cells should have enough spatial resources to transmit and receive
their desired signals and to eliminate the ICIs between these BSs and
users.


\subsection{Proof of \eqref{Eq:Necessary_Condition_Compress}}
\begin{proof}
To express the ICIs generated from the BSs in any cells to the users
in any other cells, we consider two non-overlapping clusters A and
B, as shown in Fig. \ref{fig:Incompressible_ICI}. We use
$\mathcal{I}_{\mathrm{A}}$ and $\mathcal{I}_{\mathrm{B}}$ to denote
the cell index sets in the clusters A and B, respectively, then
$\mathcal{I}_{\mathrm{A}}\cap\mathcal{I}_{\mathrm{B}}=\varnothing$.
Let $\mathcal{A}=\{j|j\in \mathcal{I}_{\mathrm{A}}\}$ and
$\mathcal{B}=\{i_k|k\in \mathcal{K}_i,~i\in
\mathcal{I}_{\mathrm{B}}\}$ denote the BS index set in cluster A and
the user index set in cluster B, respectively. The ZF
constraints to eliminate the ICI from the BSs in cluster A to the
users in cluster B can be written as
\begin{align}\label{Eq:Constraint_ICI_free_Multiple}
\pmb{U}_{\mathrm{B}}^{H}\pmb{H}_{\mathrm{B},\mathrm{A}}\pmb{V}_{\mathrm{A}}
=\pmb{0}
\end{align}
where $\pmb{V}_{\mathrm{A}}=\mathrm{diag}\{\pmb{V}_{A_{1}},\cdots,\pmb{V}_{A_{p}}\}\in \mathbb{C}^{d_{\mathrm{A}} \times
M_{\mathrm{A}}}$, $\pmb{U}_{\mathrm{B}}=\mathrm{diag}\{\pmb{U}_{B_{1}},\cdots$, $\pmb{U}_{B_{q}}\}\in
\mathbb{C}^{d_{\mathrm{B}} \times N_{\mathrm{B}}}$,
\begin{align*}
\pmb{H}_{\mathrm{B},\mathrm{A}}=
\left[\begin{array}{ccc}
        \pmb{H}_{B_{1},A_{1}} & \cdots & \pmb{H}_{B_{1},A_{p}} \\
        \vdots & \ddots & \vdots \\
        \pmb{H}_{B_{q},A_{1}} & \cdots & \pmb{H}_{B_{q},A_{p}}
      \end{array}\right]\in
\mathbb{C}^{N_{\mathrm{B}} \times M_{\mathrm{A}}}
\end{align*}
is the stacked channel
matrix from the BSs in cluster A to the users in cluster B, $A_{s}$ and $B_{s}$ are the $s$th elements in $\mathcal{A}$ and
$\mathcal{B}$, respectively,
$p\triangleq|\mathcal{A}|=|\mathcal{I}_{\mathrm{A}}| $ and
$q\triangleq|\mathcal{B}|=\sum_{i\in \mathcal{I}_{\mathrm{B}}}|\mathcal{K}_i|  $,
$d_{\mathrm{A}}=\sum_{j\in\mathcal{I}_{\mathrm{A}}}d_j$ is the
number of all data streams transmitted from the BSs in cluster A,
$d_{\mathrm{B}}=\sum_{i\in \mathcal{I}_{\mathrm{B}}}\sum_{k\in
\mathcal{K}_i}d_{i_k}$ is the number of all data streams received at
the users in cluster B,
$M_{\mathrm{A}}=\sum_{j\in\mathcal{I}_{\mathrm{A}}}M_j$ is the
number of all transmit antennas at the BSs in cluster A, and
$N_{\mathrm{B}}=\sum_{i\in \mathcal{I}_{\mathrm{B}}}\sum_{k\in
\mathcal{K}_i}N_{i_k}$ is the number of all receive antennas at the
users in cluster B.

\begin{figure}[htb!]
\centering
\includegraphics[width=0.99\linewidth]{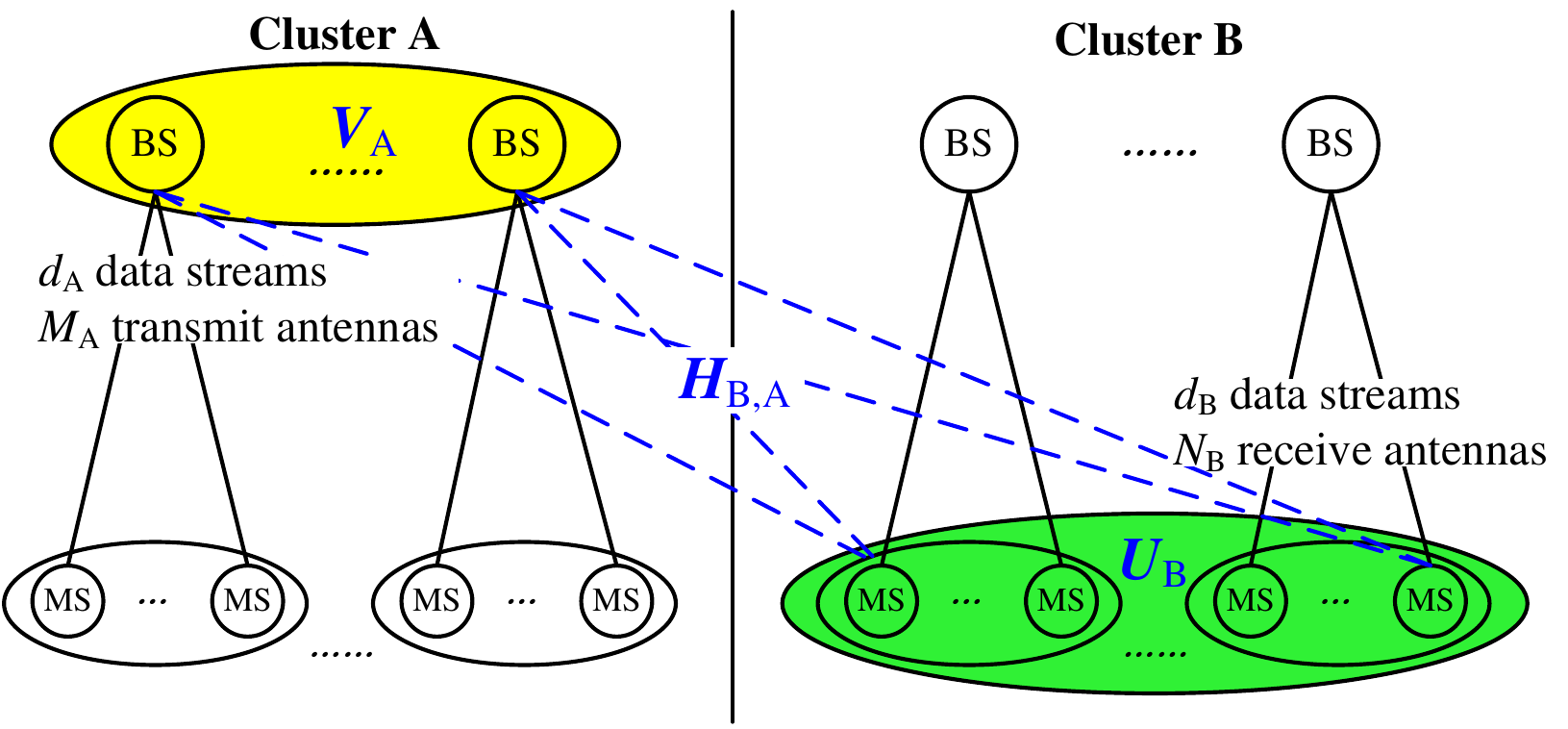}
\caption{ICI between arbitrary two non-overlapping clusters.} \label{fig:Incompressible_ICI}
\end{figure}

Since $\pmb{H}_{\mathrm{B},\mathrm{A}}$ is generic, its rank
satisfies
$\mathrm{rank}\left(\pmb{H}_{\mathrm{B},\mathrm{A}}\right)=
\min\{M_{\mathrm{A}},N_{\mathrm{B}}\}$ with probability one
\cite{Gomadam2008}. If $N_{\mathrm{B}} \geq M_{\mathrm{A}}$,
$\mathrm{rank}\left(\pmb{H}_{\mathrm{B},\mathrm{A}}\right)=
M_{\mathrm{A}} \geq d_{\mathrm{A}}$. Since
$\pmb{H}_{\mathrm{B},\mathrm{A}}$ is independent of
$\pmb{V}_{\mathrm{A}}$, we have
$\mathrm{rank}(\pmb{H}_{\mathrm{B},\mathrm{A}}\pmb{V}_{\mathrm{A}})=\mathrm{rank}(\pmb{V}_{\mathrm{A}})
=d_{\mathrm{A}}$ with probability one. Then, the users in cluster B
will see $d_{\mathrm{A}}$ ICIs from the BSs in cluster A. On the
other hand, the users in cluster B need to receive overall
$d_{\mathrm{B}}$ desired signals from the BSs in cluster B, i.e.,
$\mathrm{rank}(\pmb{U}_{\mathrm{B}})=d_{\mathrm{B}}$. To separate
the ICIs from cluster A and the desired signals of cluster B, the
overall subspace dimension of the received signals for the users in
cluster B should satisfy $N_{\mathrm{B}}\geq
d_{\mathrm{A}}+d_{\mathrm{B}}$ according to the rank-nullity theorem
\cite{Rank_nullity}.

Similarly, if $N_{\mathrm{B}} \leq M_{\mathrm{A}}$, we have
$\mathrm{rank}(\pmb{U}_{\mathrm{B}}^{H}\pmb{H}_{\mathrm{B},\mathrm{A}})=
\mathrm{rank}(\pmb{U}_{\mathrm{B}})=d_{\mathrm{B}}$ with probability
one. Then, the BSs in cluster A need to avoid generating
$d_{\mathrm{B}}$ ICIs to the users in cluster B. To transmit
$d_{\mathrm{A}}$ desired data streams, the overall subspace
dimension of the transmit signals from the BSs in cluster A should
satisfy $M_{\mathrm{A}}\geq d_{\mathrm{A}}+d_{\mathrm{B}}$.

As a result, we obtain $\max\{M_{\mathrm{A}},N_{\mathrm{B}}\}\geq
d_{\mathrm{A}}+d_{\mathrm{B}}$, i.e.,
\eqref{Eq:Necessary_Condition_Compress}.
\end{proof}

The intuitive meaning of
\eqref{Eq:Necessary_Condition_Compress} is to ensure that a sort of \emph{irreducible ICI} can be eliminated. The concept of
\emph{irreducible ICI} is explained as follows. For the ICIs between a BS in cluster A and a
set of users in cluster B, if the dimension of these ICIs cannot be
reduced by designing the receive matrices of the  users, they are
\emph{irreducible ICIs} and can only be eliminated by
the BS. Similarly, for the ICIs between several BSs in cluster A and
a user in cluster B, if the dimension of these ICIs cannot be
reduced by designing the transmit matrices of the BSs, they are
\emph{irreducible ICIs} and can only be eliminated by
the user.

From the proof of \eqref{Eq:Necessary_Condition_Compress}, we know
that when
\begin{align}\label{Eq:Inreducable_ICI}
M_j \geq \sum_{i\in \mathcal{I}_{\mathrm{B}}}\sum_{k\in
\mathcal{K}_i}N_{i_k}, \quad ~\exists \mathcal{I}_{\mathrm{B}}\subseteq
\{1,\cdots,G\}\setminus\{j\}
\end{align}
$\mathrm{rank}(\pmb{U}_{\mathrm{B}}^{H}\pmb{H}_{\mathrm{B},\mathrm{A}})=\mathrm{rank}(\pmb{U}_{\mathrm{B}})=
d_{\mathrm{B}}$ always holds. It implies that when there exists a BS
in cluster A whose dimension of observation space is no less than
the overall dimension of observation space at all the users in
cluster B, there exist the \emph{irreducible ICIs} that can only be
removed by the BS.

On the other hand, when
\begin{align}\label{Eq:Inreducable_ICI1}
N_{i_k}\geq
\sum_{j\in\mathcal{I}_{\mathrm{A}}}M_j,~\exists
\mathcal{I}_{\mathrm{A}}\subseteq \{1,\cdots,G\}\setminus\{i\}
\end{align}
$\mathrm{rank}(\pmb{H}_{\mathrm{B},\mathrm{A}}\pmb{V}_{\mathrm{A}})=d_{\mathrm{A}}$
always holds, i.e., there exist the \emph{irreducible ICIs} that can
only be removed by the user.

Therefore, we call \eqref{Eq:Inreducable_ICI} and
\eqref{Eq:Inreducable_ICI1} as the \emph{existence conditions of the
irreducible ICIs}.

To understand the impact of the irreducible ICIs, we take the case
satisfying \eqref{Eq:Inreducable_ICI} as an example, where the ICIs between BS$_j$ and
the users in cluster B are irreducible. In other words, the receive matrices
of the users in cluster B are not able to compress these ICIs. To ensure the
linear IA to be feasible, in this case
\eqref{Eq:Necessary_Condition_Compress} requires $M_j\geq
d_j+\sum_{i\in \mathcal{I}_{\mathrm{B}}}\sum_{k\in
\mathcal{K}_i}d_{i_k}$, which is equivalent to
\begin{align}\label{Eq:Imcompressible_ICI}
  (M_j-d_j)d_j\geq d_j\sum_{i\in
\mathcal{I}_{\mathrm{B}}}\sum_{k\in
\mathcal{K}_i}d_{i_k}
\end{align}
It means that the number of variables in the effective transmit
matrix at BS$_j$ should exceed the number of ICIs. In other words,
BS$_j$ should be able to avoid the ICI. Therefore,
\eqref{Eq:Imcompressible_ICI} is \emph{the condition of eliminating
the irreducible ICI}.

By contrast, if \eqref{Eq:Inreducable_ICI} does not hold, these ICIs
are reducible at the BS passively, because anyway the BS only
``see'' these ICIs in a subspace with lower dimension of $M_j$ than
the overall observation space at all the users in cluster B with
dimension of $\sum_{i\in \mathcal{I}_{\mathrm{B}}}\sum_{k\in
\mathcal{K}_i}N_{i_k}$. In this case, the ICIs between BS$_j$ and
the users in cluster B can be removed by their implicit
``cooperation'' of sharing variables in their processing matrices.
To eliminate the ICI from BS$_j$ to the users in cluster B, in this
case the proper condition
\eqref{Eq:Necessary_Condition_Interference} requires
\begin{align}\label{Eq:Imcompressible_ICI_1}
  (M_j-d_j)d_j + \sum_{i\in
\mathcal{I}_{\mathrm{B}}}\sum_{k\in
\mathcal{K}_i}(N_{i_k}-d_{i_k})d_{i_k}\geq d_j\sum_{i\in
\mathcal{I}_{\mathrm{B}}}\sum_{k\in \mathcal{K}_i}d_{i_k}
\end{align}
It indicates that the overall number of variables in the transmit
and receive matrices at both BS$_j$ and the users in set
$\mathcal{K}_i$ should exceed the overall number of ICIs among them.
Therefore, \eqref{Eq:Imcompressible_ICI_1} is \emph{the condition of
eliminating the reducible ICI}.

In practice, there exist both the \emph{reducible ICI} and the
\emph{irreducible ICI} in a MIMO-IBC. Comparing
\eqref{Eq:Imcompressible_ICI} and \eqref{Eq:Imcompressible_ICI_1},
we can see that the proper conditions only ensure to eliminate all
the reducible ICIs but not all the irreducible ICIs. This explains
the reason why the IA is infeasible for the system whose
configuration satisfies \eqref{Eq:Necessary_Condition_Interference}
but does not satisfy \eqref{Eq:Necessary_Condition_Compress}.

In summary, \eqref{Eq:Necessary_Condition_Signal} ensures that there
are enough antennas to convey the \emph{desired signals} between
each BS and each user. \eqref{Eq:Necessary_Condition_Interference}
ensures to eliminate the \emph{reducible ICIs} by sharing the
spatial resources between the BSs and the users, whereas
\eqref{Eq:Necessary_Condition_Compress} ensures to eliminate the
\emph{irreducible ICIs} either at the BS side or at the user side.

\section{Necessary and Sufficient Conditions for a Special Class of MIMO-IBC}\label{Sec:Sufficient_Condition}
In this section, we present and prove the \emph{necessary and
sufficient conditions} of linear IA feasibility for a special class
of MIMO-IBC, where the numbers of transmit and receive antennas are
all divisible by the number of data streams of each user. Owing to
the second feature of MIMO-IBC, the sufficiency proof for MIMO-IBC
is more difficult than the special class of MIMO-IC in
\cite{Luo2012}.

We start by proving the necessity, which is simple. Then, we analyze
two important properties of the Jacobian matrix for general MIMO-IBC
and MIMO-IC. We proceed to present three lemmas to show the impact
of the two properties. Finally, we prove the sufficiency by
constructing an invertible Jacobian matrix for the considered
MIMO-IBC, i.e., find the minimal antenna configuration to ensure the
IA feasibility.

\begin{theorem}[Necessary and Sufficient
Conditions]\label{Theorem:NS_conditions} For a special class of MIMO-IBC with configuration
$\prod_{i=1}^{G}(M_i\times \prod_{k=1}^{K_i}(N_{i_k},d_{i_k}))$
where the channel matrices $\{\pmb{H}_{i_k,j}\}$ are generic, when $d_{i_k}=d$, and both $M_i$
and $N_{i_k}$ are divisible by $d$, the linear IA is feasible iff (if and only if) the following
conditions are satisfied,
\begin{subequations}
\begin{align}
\label{Eq:Sufficient_Condition_Signal}
&\min\{M_i-K_id,N_{i_k}-d\}\geq 0, ~\forall i,k\\
\label{Eq:Sufficient_Condition_Interference}
&\sum_{j:(i,j)\in \mathcal{I}}
(M_j-K_j d)K_j+ \sum_{i:(i,j)\in \mathcal{I}}\sum_{k\in
\mathcal{K}_i}(N_{i_k}-d) \nonumber\\
\geq&
\sum_{(i,j)\in\mathcal{I}}K_j|\mathcal{K}_i|d,~\forall
\mathcal{I} \subseteq \mathcal{J}
\end{align}
\end{subequations}
\end{theorem}

\subsection{Proof of the necessity}
\begin{proof}
Comparing \emph{Theorem \ref{Theorem:Necessary_Condition}} and
\emph{Theorem \ref{Theorem:NS_conditions}}, we can see that
\eqref{Eq:Sufficient_Condition_Signal} and
\eqref{Eq:Sufficient_Condition_Interference} are the reduced forms of
\eqref{Eq:Necessary_Condition_Signal} and
\eqref{Eq:Necessary_Condition_Interference} for the special class of
MIMO-IBC in \emph{Theorem \ref{Theorem:NS_conditions}}. For this class of MIMO-IBC,
\eqref{Eq:Necessary_Condition_Compress} becomes
\begin{align}\label{Eq:Necessary_Condition_Compress_d}
&\max\Big\{\sum_{j\in\mathcal{I}_{\mathrm{A}}}M_j,\sum_{i\in
\mathcal{I}_{\mathrm{B}}}\sum_{k\in
\mathcal{K}_i}N_{i_k}\Big\}\nonumber\\
\geq&
\sum_{j\in\mathcal{I}_{\mathrm{A}}}K_jd+
\sum_{i\in\mathcal{I}_{\mathrm{B}}}|\mathcal{K}_i|d,~\forall \mathcal{I}_{\mathrm{A}}\cap\mathcal{I}_{\mathrm{B}}=\varnothing
\end{align}
In the sequel, we show that
\eqref{Eq:Necessary_Condition_Compress_d} can be derived from
\eqref{Eq:Sufficient_Condition_Signal} and
\eqref{Eq:Sufficient_Condition_Interference}.

Since $M_j$ is integral multiples of $d$, the value of $M_j$ can be
divided into two cases:
\begin{enumerate}
  \item $\sum_{j\in\mathcal{I}_{\mathrm{A}}}M_j\geq
(\sum_{j\in\mathcal{I}_{\mathrm{A}}}K_j+
\sum_{i\in\mathcal{I}_{\mathrm{B}}}|\mathcal{K}_i|)d$,
  \item $\sum_{j\in\mathcal{I}_{\mathrm{A}}}M_j\leq
(\sum_{j\in\mathcal{I}_{\mathrm{A}}}K_j+
\sum_{i\in\mathcal{I}_{\mathrm{B}}}|\mathcal{K}_i|-1)d$.
\end{enumerate}
In the first case, \eqref{Eq:Necessary_Condition_Compress_d} always
holds. In the second case, we have
\begin{align}\label{condi13}
\sum_{i\in\mathcal{I}_{\mathrm{B}}}|\mathcal{K}_i|d
-\sum_{j\in\mathcal{I}_{\mathrm{A}}}(M_j-K_jd)\geq d
\end{align}

Considering \eqref{Eq:Sufficient_Condition_Signal}, we know that
$M_j-K_jd \geq 0$. Thus the inequality
$\sum_{j\in\mathcal{I}_{\mathrm{A}}}(M_j-K_jd) \geq M_j-K_jd$ always
holds. Substituting this inequality into \eqref{condi13}, we have
\begin{align}\label{condi14}
\sum_{i\in\mathcal{I}_{\mathrm{B}}}|\mathcal{K}_i|d
-(M_j-K_jd)\geq d
\end{align}

Considering the definition of $\mathcal{I}_{\mathrm{A}}$ and $\mathcal{I}_{\mathrm{B}}$ in \eqref{Eq:Necessary_Condition_Compress_d}, \eqref{Eq:Sufficient_Condition_Interference} can be rewritten as $\sum_{j\in
\mathcal{I}_{\mathrm{A}}}
(M_j-K_jd)K_j+ \sum_{i\in\mathcal{I}_{\mathrm{B}}}\sum_{k\in
\mathcal{K}_i}(N_{i_k}-d)\nonumber\\
\geq
\sum_{j\in\mathcal{I}_{\mathrm{A}}}K_j \sum_{i\in\mathcal{I}_{\mathrm{B}}}|\mathcal{K}_i|d$,
which is equivalent to
\begin{align}\label{Eq:Necessary_Condition_Compress_d_1}
\sum_{i\in
\mathcal{I}_{\mathrm{B}}}\sum_{k\in
\mathcal{K}_i}N_{i_k} \geq &
\sum_{j\in\mathcal{I}_{\mathrm{A}}}K_j \Big(\sum_{i\in\mathcal{I}_{\mathrm{B}}}|\mathcal{K}_i|d-(M_j-K_jd)\Big)
\nonumber\\
&+\sum_{i\in\mathcal{I}_{\mathrm{B}}}|\mathcal{K}_i|d
\end{align}

Substituting \eqref{condi14} into
\eqref{Eq:Necessary_Condition_Compress_d_1}, we obtain
\eqref{Eq:Necessary_Condition_Compress_d}.
\end{proof}

For the considered class of MIMO-IBC, since
\eqref{Eq:Necessary_Condition_Compress_d} (i.e.,
\eqref{Eq:Necessary_Condition_Compress}) can be derived from
\eqref{Eq:Sufficient_Condition_Interference} (i.e.,
\eqref{Eq:Necessary_Condition_Interference}), the proper condition
ensures that when there exist some \emph{irreducible ICIs}, the BS
or the user has enough spatial resources to avoid (or cancel) them.

\subsection{Proof of the sufficiency}
From the analysis in \cite{Luo2012,Tse2011}, we know that the linear
IA will be feasible for general MIMO-IC and MIMO-IBC under
\emph{generic} channels if we can find a channel realization that
has an IA solution and whose Jacobian matrix is invertible.

Consider a channel matrix as follows
\begin{align*}
\bar{\pmb{H}}_{0i_k,j}=\left[\begin{array}{cc}
\pmb{0}& \bar{\pmb{H}}_{0~i_k,j}^{(2)}\\
\bar{\pmb{H}}_{0~i_k,j}^{(3)}& \pmb{0}
\end{array}\right]
\end{align*}
under which an IA solution can be easily found as
\begin{align*}
{\pmb{V}}_{0j}= \left[\begin{array}{c}
\pmb{I}_{d_j}\\
\pmb{0}_{(M_j-d_j)\times d_j}
\end{array}\right],~~{\pmb{U}}_{0i_k}= \left[\begin{array}{c}
\pmb{I}_{d_{i_k}}\\
\pmb{0}_{(N_{i_k}-d_{i_k})\times d_{i_k}}
\end{array}\right]
\end{align*}
Then, to prove the sufficiency, we only need to construct a Jacobian matrix that is invertible at
$\bar{\pmb{H}}_{0i_k,j}$.

Substituting $\bar{\pmb{H}}_{0~i_k,j}$ into \eqref{Eq:Constraint_ICI_Transmission_free}, we obtain
\begin{align}\label{Eq:Constraint_ICI_free_Linear}
\pmb{F}_{i_k,j}(\bar{\pmb{H}}_{0})=\bar{\pmb{H}}_{0~i_k,j}^{(2)}\bar{\pmb{V}}_{j}+ \bar{\pmb{U}}_{i_k}^{H}\bar{\pmb{H}}_{0~i_k,j}^{(3)}
=\pmb{0},~\forall i\neq j
\end{align}
By taking partial derivatives to
\eqref{Eq:Constraint_ICI_free_Linear}, we can obtain the Jacobian
matrix of $\bar{\pmb{H}}_{0}$.\footnote{Since
\eqref{Eq:Constraint_ICI_free_Linear} is a system represented by linear
polynomials, its first-order coefficients are its partial
derivatives. The condition that the first-order
coefficients of polynomials are linear independent
\cite{Tse2011,Win_IC} is the same as the condition that the Jacobian
matrix is invertible \cite{Luo2012}.} Before constructing an
invertible Jacobian matrix, we first analyze its properties.

\subsubsection{Jacobian matrix: properties and impacts}
To see the structure of the Jacobian matrices for general MIMO-IC
and MIMO-IBC, we rewrite the matrices in
\eqref{Eq:Constraint_ICI_free_Linear} as
$\pmb{F}_{i_k,j}(\bar{\pmb{H}}_{0})=[\pmb{F}_{i_k,j_1}(\bar{\pmb{H}}_{0}),
\cdots,\pmb{F}_{i_k,j_{K_j}}(\bar{\pmb{H}}_{0})]$,
$\bar{\pmb{V}}_j=[\bar{\pmb{V}}_{j_1},\cdots,\bar{\pmb{V}}_{j_{K_j}}]$,
and $\bar{\pmb{H}}_{0~i_k,j}^{(3)}=[\bar{\pmb{H}}_{0~i_k,j_1}^{(3)},
\cdots,\bar{\pmb{H}}_{0~i_k,j_{K_j}}^{(3)}]$, where
$\pmb{F}_{i_k,j_l}(\bar{\pmb{H}}_{0})\in \mathbb{C}^{d_{i_k}\times d_{j_l}}$,
$\bar{\pmb{V}}_{j_l}\in \mathbb{C}^{(M_j-d_j)\times d_{j_l}}$, and
$\bar{\pmb{H}}_{0~i_k,j_l}^{(3)}\in
\mathbb{C}^{(N_{i_k}-d_{i_k})\times d_{j_l}}$. Then
\eqref{Eq:Constraint_ICI_free_Linear} can be rewritten as $K_j$ groups of
subequations, where the $l$th subequation is
\begin{align}\label{Eq:Constraint_ICI_free_New5}
&\pmb{F}_{i_k,j_l}(\bar{\pmb{H}}_{0})
=\bar{\pmb{H}}_{0~i_k,j}^{(2)}\bar{\pmb{V}}_{j_l} +\bar{\pmb{U}}_{i_k}^{H}\bar{\pmb{H}}_{0~i_k,j_l}^{(3)}=\pmb{0},~\forall i\neq j
\end{align}

The Jacobian matrix of the polynomial map \eqref{Eq:Constraint_ICI_free_New5} is
\begin{align}\label{Eq:Jacobian_Matrix_Define}
  \pmb{J}\triangleq \left[\frac{\partial\mathrm{vec}\{\pmb{F}\}}{\partial\mathrm{vec}\{\bar{\pmb{V}};\bar{\pmb{U}}\}}\right] =[\pmb{J}^{V},\pmb{J}^{U}]
\end{align}
where $\pmb{J}^{V} =$ $ \partial\mathrm{vec}\{\pmb{F}\}/\partial\mathrm{vec}\{\bar{\pmb{V}}\}$, $\pmb{J}^{U} = $ $ \partial\mathrm{vec}\{\pmb{F}\}/\partial\mathrm{vec}\{\bar{\pmb{U}}\}$, $\mathrm{vec}\{\bar{\pmb{V}}\}=$ $
[\mathrm{vec}\{\bar{\pmb{V}}_{1_1}\}^{T},\cdots,\mathrm{vec}\{\bar{\pmb{V}}_{G_{K_G}}\}^{T}]^{T}$,
and $\mathrm{vec}\{\bar{\pmb{U}}\}= $ $ [\mathrm{vec}\{\bar{\pmb{U}}^{H}_{1_1}\}^{T}, \cdots,\mathrm{vec}\{\bar{\pmb{U}}^{H}_{G_{K_G}}\}^{T}]^{T}$.

Substituting \eqref{Eq:Constraint_ICI_free_New5} into
\eqref{Eq:Jacobian_Matrix_Define}, the elements of $\pmb{J}(\bar{\pmb{H}}_{0})$ are
\begin{subequations}
\begin{align}\label{Eq:Jacobian_Matrix_IBC_V}
\frac{\partial\mathrm{vec}\{\pmb{F}_{i_k,j_l}(\bar{\pmb{H}}_{0})\}}
  {\partial\mathrm{vec}\{\bar{\pmb{V}}_{m_n}\}}&=
\left\{
\begin{array}{ll}
\bar{\pmb{H}}_{0~i_k,j}^{(2)}\otimes\pmb{I}_{d_{j_l}},&~\forall m_n=j_l\\
\pmb{0}_{d_{i_k}d_{j_l} \times (M_m-d_m)d_{m_n}},&~\forall m_n\neq j_l
\end{array}
  \right.
  \\
\label{Eq:Jacobian_Matrix_IBC_U}
  \frac{\partial\mathrm{vec}\{\pmb{F}_{i_k,j_l}(\bar{\pmb{H}}_{0})\}}
  {\partial\mathrm{vec}\{\bar{\pmb{U}}^{H}_{m_n}\}}&=
\left\{
\begin{array}{ll}
\pmb{I}_{d_{i_k}} \otimes (\bar{\pmb{H}}_{0~i_k,j_l}^{(3)})^{T},  &~\forall m_n=i_k\\
\pmb{0}_{d_{i_k}d_{j_l} \times (N_{m_n}-d_{m_n})d_{m_n}},&~\forall  m_n\neq i_k
\end{array}
  \right.
\end{align}
\end{subequations}
where the nonzero elements in \eqref{Eq:Jacobian_Matrix_IBC_V} satisfy
\begin{align}\label{Eq:Jacobian_Matrix_IBC_VV}
\frac{\partial\mathrm{vec}\{\pmb{F}_{i_k,j_1}(\bar{\pmb{H}}_{0})\}}
  {\partial\mathrm{vec}\{\bar{\pmb{V}}_{j_1}\}}=\cdots= \frac{\partial\mathrm{vec}\{\pmb{F}_{i_k,j_{K_j}}(\bar{\pmb{H}}_{0})\}}
  {\partial\mathrm{vec}\{\bar{\pmb{V}}_{j_{K_j}}\}}
\end{align}

We can see that the Jacobian
matrix of general MIMO-IC and MIMO-IBC has two properties in structure:
\begin{itemize}
  \item \emph{Sparse structure}: There are many zero blocks in particular positions as shown in \eqref{Eq:Jacobian_Matrix_IBC_V} and
\eqref{Eq:Jacobian_Matrix_IBC_U} \cite{Luo2012,Win_IC,Gonzalez_IC}.
  \item \emph{Repeated structure}: There are repeated nonzero elements in particular positions when $d_{i_k}>1$ (i.e., multi-beam MIMO-IC \cite{Win_IC,Gonzalez_IC}
  or MIMO-IBC) or $K_i>1$ (i.e., MIMO-IBC), $\exists i,k$, as shown in \eqref{Eq:Jacobian_Matrix_IBC_VV}.
\end{itemize}

Such a \emph{repeated structure} comes from the second feature of MIMO-IBC.
Comparing \eqref{Eq:Jacobian_Matrix_IBC_V}, \eqref{Eq:Jacobian_Matrix_IBC_U} with \eqref{Eq:Jacobian_Matrix_IBC_VV}, we can see that the repeated elements caused by multi-beam will appear in both $\pmb{J}^{V}$ and $\pmb{J}^{U}$, while the repeated elements caused by multi-user only appear in $\pmb{J}^{V}$ but not $\pmb{J}^{U}$. Therefore, the repeated structure of the Jacobian matrix for MIMO-IBC is quite different with that for multi-beam MIMO-IC.

In the sequel, we introduce three lemmas to show the impact of the two properties on constructing an invertible Jacobian matrix for MIMO-IBC.\footnote{Because when $K_i=1$ MIMO-IBC reduces to MIMO-IC, the conclusions for MIMO-IBC are also valid for MIMO-IC.}

It is worth to note that an invertible Jacobian matrix in the case of $L_v>L_e$ can be
obtained from that in the case of $L_v=L_e$, since one
can always remove some redundant variables to ensure $L_v=L_e$, where $L_v$ and $L_e$
denote the total numbers of scalar variables and equations in \eqref{Eq:Constraint_ICI_free_New5}. Therefore, we only need to investigate the case of $L_v=L_e$.

\begin{lemma}\label{Lemma:Sparse}
For a proper MIMO-IBC with $L_v=L_e$, there always exists a
permutation matrix that has the same \emph{sparse structure} as the
Jacobian matrix, and the permutation matrix can be obtained from a
perfect matching in a bipartite representing
\eqref{Eq:Constraint_ICI_free_New5}.\footnote{The relationship
between the equations and variables in
\eqref{Eq:Constraint_ICI_free_New5} can be represented by a
bipartite graph, where a set of non-adjacent edges is called a
\emph{matching}. If a matching matches all vertices of the graph, it
is called a \emph{perfect matching} \cite{Grahp_Theory}. The perfect
matching is first used to construct invertible Jacobian matrices in
\cite{Jocobian_Graph} and first used to solve the IA feasibility for
MIMO-IC in \cite{Luo2012}. }
\end{lemma}

\begin{proof}
See Appendix \ref{App_Lemma3}.
\end{proof}

Considering that in general the Jacobian matrix for MIMO-IBC has
both the \emph{sparse structure} and the \emph{repeated structure},
a permutation matrix that has the same \emph{sparse structure} as
the Jacobian matrix is not necessarily a Jacobian matrix. One
exception is the single beam MIMO-IC, where $K_i=1$, $d_{i}=1$. This
is because its Jacobian matrix does not have the \emph{repeated
structure}, which can be set as a permutation matrix from
\emph{Lemma \ref{Lemma:Sparse}}.

\begin{lemma}\label{Lemma:Repeated}
For the special class of MIMO-IBC in \emph{Theorem
\ref{Theorem:NS_conditions}}, an invertible Jacobian matrix for a
multi-beam MIMO-IBC can be constructed from a single beam MIMO-IBC.
Moreover, if the Jacobian matrix for this class of systems with
$d=1$ is a permutation matrix, the Jacobian matrix for the systems
with $d>1$ will also be a permutation matrix.
\end{lemma}

\begin{proof}
See Appendix \ref{App_Lemma1}.
\end{proof}

The Jacobian matrix of multi-beam MIMO-IC has the \emph{repeated
structure}, and the Jacobian matrix of single beam MIMO-IC can be
set as a permutation matrix. \emph{Lemma \ref{Lemma:Repeated}}
implies that for a class of multi-beam MIMO-IC where each user
expects $d$ data streams and both the transmit and receive antennas
are divisible by $d$, there exists a permutation matrix that can
satisfy the two properties of the Jacobian matrix simultaneously.
Consequently, the Jacobian matrix of this class of MIMO-IC can be
set as a permutation matrix as shown in \cite{Luo2012}.

\begin{lemma}\label{Lemma:CanNot_Extended} For a class of MIMO-IBC with $L_v
= L_e$, $\sum_{j=1,j\neq i}^{G}d_j\geq N_{i_k}-d_{i_k}>0$ and
$N_{i_k}-d_{i_k}\notin \phi_{i}$, $\exists i,k$, where $\phi_{i}=\{\sum_{j\in
\psi_i}d_j |\psi_i \subseteq \{1,\cdots,G\}\setminus\{i\}\}$, when
\eqref{Eq:Sufficient_Condition_Signal} and
\eqref{Eq:Sufficient_Condition_Interference} are satisfied, there
does not exist any Jacobian matrix that is a permutation matrix.
\end{lemma}

\begin{proof}
See Appendix \ref{App_Lemma5}.
\end{proof}

For the class of MIMO-IBC in \emph{Lemma
\ref{Lemma:CanNot_Extended}}, whose Jacobian matrix has the
\emph{repeated structure}, one cannot find a permutation matrix that
satisfies both the two properties of the Jacobian matrix. The
sufficiency of IA feasibility for this class of MIMO-IBC has not
been proved only with one exception in \cite{Tse2011} as shown
later, where constructing an invertible Jacobian matrix is difficult
due to the confliction of its two properties.

A sub-class of MIMO-IBC considered in the lemma is also considered
by \emph{Theorem \ref{Theorem:NS_conditions}}. We show this with
several examples. In \emph{Lemma \ref{Lemma:CanNot_Extended}},
$\phi_i$ is a set of data stream number in one or multiple cells,
whose desired signals will generate ICI to the users in cell $i$.
When $G=3$, $\phi_{1}=\{d_2,d_3,(d_2+d_3)\}$,
$\phi_{2}=\{d_1,d_3,(d_1+d_3)\}$ and
$\phi_{3}=\{d_1,d_2,(d_1+d_2)\}$. For a symmetric MIMO-IBC with
configuration $(M\times (N,d)^{K})^{G}$, we have $d_1=\cdots=d_G=Kd$
and $\phi_{1}=\cdots=\phi_{G}=\{Kd,2Kd,\cdots,(G-1)Kd\}$, then the
condition that $N_{i_k}-d_{i_k}\notin \phi_i$ reduces to the
condition that $N-d$ is not divisible by $Kd$. When $K=1,d>1$, the
system becomes a symmetric multi-beam MIMO-IC and the condition
reduces to that $N$ is not divisible by $d$. In this case, the
sufficiency was only proved for a multi-beam MIMO-IC with $M=N$ in
\cite{Tse2011}. When $K>1,d=1$, the system becomes a symmetric
single beam MIMO-IBC and the condition reduces to that $N-1$ is not
divisible by $K$, which is one case of those considered in
\emph{Theorem \ref{Theorem:NS_conditions}}.

\subsubsection{Proof of the sufficiency in Theorem
\ref{Theorem:NS_conditions}}~~~~~~

\begin{proof}
According to \emph{Lemma \ref{Lemma:Repeated}}, in the following we only need to
construct an invertible Jacobian matrix for a corresponding
single beam MIMO-IBC, i.e., the case where $d = 1$.

When $d=1$, the effective transmit and receive matrices
$\bar{\pmb{V}}_{j_l}$ and $\bar{\pmb{U}}_{i_k}$ defined in
\eqref{effective txrx} reduce to the effective transmit and receive
vectors $\bar{\pmb{v}}_{j_l}$ and $\bar{\pmb{u}}_{i_k}$. Then,
\eqref{Eq:Constraint_ICI_free_New5} is simplified as
\begin{align}\label{Eq:Constraint_ICI_free_SingleStream}
&F_{i_k,j_l}(\bar{\pmb{H}}_{0})=\bar{\pmb{h}}_{0~i_k,j}^{(2)}\bar{\pmb{v}}_{j_l}+
\bar{\pmb{u}}_{i_k}^{H}\bar{\pmb{h}}_{0~i_k,j_l}^{(3)}=0,~\forall i\neq j
\end{align}

To construct an invertible Jacobian matrix for the MIMO-IBC with $d=1$, we first analyze the structure of the
matrix.

In the Jacobian matrix, the
rows correspond to the entries of equation set
$\mathcal{Y} \triangleq \{F_{i_k,j_l}(\bar{\pmb{H}}_{0})|\forall i\neq j\}$, and the columns
correspond to the entries of variable set.
To show its \emph{repeated structure}, we divide all the
ICIs into $\sum_{j=1}^{G}K_j$ subsets, where
$\mathcal{Y}_{j_l} \triangleq \{F_{i_k,j_l}(\bar{\pmb{H}}_{0})|k=1,\cdots,K_i,~i=1,\cdots,G, \forall i\neq
j\}$ is a subset of the ICIs generated from the effective transmit
vector $\bar{\pmb{v}}_{j_l}$ of BS$_j$ to all the users in other cells, which contains $\sum_{i=1,i\neq j}^{G}K_i$ ICIs. Since $\mathcal{Y}=\cup_{j=1}^{G}\cup_{l=1}^{K_j}\mathcal{Y}_{j_l}$,
the overall number of ICIs is $\sum_{j=1}^{G}\sum_{l=1}^{K_j}\sum_{i=1,i\neq
j}^{G}K_i=\sum_{j=1}^{G}\sum_{i=1,i\neq
j}^{G} K_jK_i$. Consequently, the Jacobian matrix will be invertible if
\begin{align}\label{Eq:Rank_d1}
\mathrm{rank}(\pmb{J}(\bar{\pmb{H}}_0))=\sum_{j=1}^{G}\sum_{i=1,i\neq
j}^{G} K_jK_i
\end{align}

Correspondingly, $\pmb{J}(\bar{\pmb{H}}_0)$ can be
partitioned into $\sum_{j=1}^{G}K_j$ blocks, i.e.,
\begin{align}\label{Eq:Jacobian_Matrix_partitions}
\pmb{J}(\bar{\pmb{H}}_0)=[\pmb{J}_{1}^{T}(\bar{\pmb{H}}_0),\cdots,\pmb{J}_{G}^{T}(\bar{\pmb{H}}_0)]^{T}
\end{align}
where $\pmb{J}_j(\bar{\pmb{H}}_0)=[\pmb{J}_{j_1}^{T}(\bar{\pmb{H}}_0),\cdots,\pmb{J}_{j_{K_j}}^{T}(\bar{\pmb{H}}_0)]^{T}$, the rows of the $j_l$th block $\pmb{J}_{j_l}(\bar{\pmb{H}}_0)$ correspond to the
ICIs in $\mathcal{Y}_{j_l}$.

From \eqref{Eq:Jacobian_Matrix_Define} we know that $\pmb{J}_{j_l}(\bar{\pmb{H}}_0)$
can be partitioned into $\pmb{J}_{j_l}(\bar{\pmb{H}}_0) =
[\pmb{J}_{j_l}^{V}(\bar{\pmb{H}}_0),\pmb{J}_{j_l}^{U}(\bar{\pmb{H}}_0)]$, where
$\pmb{J}_{j_l}^{V}(\bar{\pmb{H}}_0)=\partial
\mathrm{vec}\{\mathcal{Y}_{j_l}\}/\partial
\mathrm{vec}\{\bar{\pmb{V}}\}$ and $\pmb{J}_{j_l}^{U}(\bar{\pmb{H}}_0)=\partial
\mathrm{vec}\{\mathcal{Y}_{j_l}\}/\partial
\mathrm{vec}\{\bar{\pmb{U}}\}$. Furthermore,
$\pmb{J}_{j_l}^{V}(\bar{\pmb{H}}_0)$ can be further divided into $\sum_{i=1}^{G}K_i$ blocks,
i.e.,
$\pmb{J}_{j_l}^{V}(\bar{\pmb{H}}_0)=[\pmb{J}_{j_l,1_1}^{V}(\bar{\pmb{H}}_{0}), \cdots,\pmb{J}_{j_l,G_{K_G}}^{V}(\bar{\pmb{H}}_{0})]$,
where $\pmb{J}_{j_l,m_n}^{V}(\bar{\pmb{H}}_0)=\partial
\mathrm{vec}\{\mathcal{Y}_{j_l}\}/\partial \bar{\pmb{v}}_{m_n}$,
whose rows correspond to all the ICIs generated from $\bar{\pmb{v}}_{j_l}$ and
columns correspond to all the variables provided by
$\bar{\pmb{v}}_{m_n}$.

According to \eqref{Eq:Jacobian_Matrix_IBC_V}, \eqref{Eq:Jacobian_Matrix_IBC_U} and
\eqref{Eq:Constraint_ICI_free_SingleStream}, the elements of the Jacobian matrix are
\begin{subequations}
\begin{align}\label{Eq:Jacobian_Matrix_IBC_D1_V}
  \frac{\partial {F}_{i_k,j_l}(\bar{\pmb{H}}_{0})}
  {\partial\bar{\pmb{v}}_{m_n}}=&
\left\{
\begin{array}{ll}
\bar{\pmb{h}}_{0~i_k,j}^{(2)},&~\forall m_n=j_l\\
\pmb{0}_{1\times (M_m-K_m)}, &~\forall m_n\neq j_l
\end{array}
  \right.
  \\
\label{Eq:Jacobian_Matrix_IBC_D1_U}
\frac{\partial{F}_{i_k,j_l}(\bar{\pmb{H}}_{0})}
  {\partial\bar{\pmb{u}}^{H}_{m_n}}=&
\left\{
\begin{array}{ll}
(\bar{\pmb{h}}_{0~i_k,j_l}^{(3)})^{T},  &~\forall m_n=i_k\\
\pmb{0}_{1\times (N_{m_n}-1)}, &~\forall m_n\neq i_k
\end{array}
  \right.
\end{align}
\end{subequations}

In \eqref{Eq:Jacobian_Matrix_IBC_D1_V} and
\eqref{Eq:Jacobian_Matrix_IBC_D1_U}, the nonzero elements respectively satisfy
\begin{subequations}
\begin{align}
\label{Eq:Jacobian_Matrix_Repetition_V1} &\frac{\partial
{F}_{i_k,j_1}(\bar{\pmb{H}}_{0})} {\partial \bar{\pmb{v}}_{j_1}}=\cdots
=\frac{\partial {F}_{i_k,j_{K_j}}(\bar{\pmb{H}}_{0})} {\partial
\bar{\pmb{v}}_{j_{K_j}}}
=\bar{\pmb{h}}_{0~i_k,j}^{(2),t}\\
\label{Eq:Jacobian_Matrix_Repetition_U1} &\frac{\partial
{F}_{i_1,j_l}(\bar{\pmb{H}}_{0})} {\partial
\bar{\pmb{u}}^{H}_{i_1}}=(\bar{\pmb{h}}_{0~i_1,j_l}^{(3)})^{T},
~\cdots,~\frac{\partial {F}_{i_{K_i},j_{l}}(\bar{\pmb{H}}_{0})}
{\partial \bar{\pmb{u}}^{H}_{i_{K_i}}} =(\bar{\pmb{h}}_{0~i_{K_i},j_l}^{(3)})^{T}
\end{align}
\end{subequations}

From \eqref{Eq:Jacobian_Matrix_IBC_D1_V}, we obtain
\begin{align}\label{Eq:Repeated_Jacobian}
  \pmb{J}_{j_l,i_k}^{V}(\bar{\pmb{H}}_0)=\left\{\begin{array}{ll}
                   \bar{\pmb{H}}_{0~:,j}^{(2)},&~\forall i_k=j_l\\
                   \pmb{0}_{\sum_{i=1,i\neq j}^{G}K_i\times (M_i-K_i)},&~\forall i_k\neq j_l
                 \end{array}
  \right.
\end{align}
where
$\bar{\pmb{H}}_{0~:,j}^{(2)}=[(\bar{\pmb{h}}_{0~1_1,j}^{(2)})^{T}$,
$\cdots$, $(\bar{\pmb{h}}_{0~{(j-1)}_{K_{(j-1)}},j}^{(2)})^{T}$,
$(\bar{\pmb{h}}_{0~{(j+1)}_{1},j}^{(2)})^{T}$, $\cdots$,
$(\bar{\pmb{h}}_{0~G_{K_G},j}^{(2)})^{T}]^{T}$.

Then, $\pmb{J}^{V}(\bar{\pmb{H}}_0)=
\mathrm{diag}\{\pmb{J}_{1}^{V}(\bar{\pmb{H}}_0),\cdots,\pmb{J}_{G}^{V}(\bar{\pmb{H}}_0)\}$,
where $\pmb{J}_{j}^{V}(\bar{\pmb{H}}_0)= \mathrm{diag}\{\pmb{J}_{j_1,j_1}^{V}(\bar{\pmb{H}}_0),\cdots,\pmb{J}_{j_{K_j},j_{K_j}}^{V}(\bar{\pmb{H}}_0)\}$. From \eqref{Eq:Jacobian_Matrix_Repetition_V1}, we have $\pmb{J}_{j_1,j_1}^{V}(\bar{\pmb{H}}_0)=\cdots=\pmb{J}^{V}_{j_{K_j},j_{K_j}}(\bar{\pmb{H}}_0)=\bar{\pmb{H}}_{0~:,j}^{(2)}$. This indicates that $\pmb{J}_{j}^{V}(\bar{\pmb{H}}_0)$ is composed of $K_j$ repeated blocks.
From \eqref{Eq:Jacobian_Matrix_Repetition_U1}, we see that the nonzero elements corresponding to the receive vectors of
the users are different, i.e., $\pmb{J}^{U}(\bar{\pmb{H}}_0)$ does not contain any repeated nonzero elements.

Figure
\ref{fig:Structure_IBC} shows the structure for the MIMO-IBC with $d = 1$, where the repeated blocks
are marked with the same kind of shadowing field, and the blank
space denotes the zero elements. We can see that in the Jacobian matrix the
blocks corresponding to the transmit vectors from each BS are
identical. This comes from the second feature of MIMO-IBC.

\begin{figure}[htb!]
\centering
\includegraphics[width=0.95\linewidth]{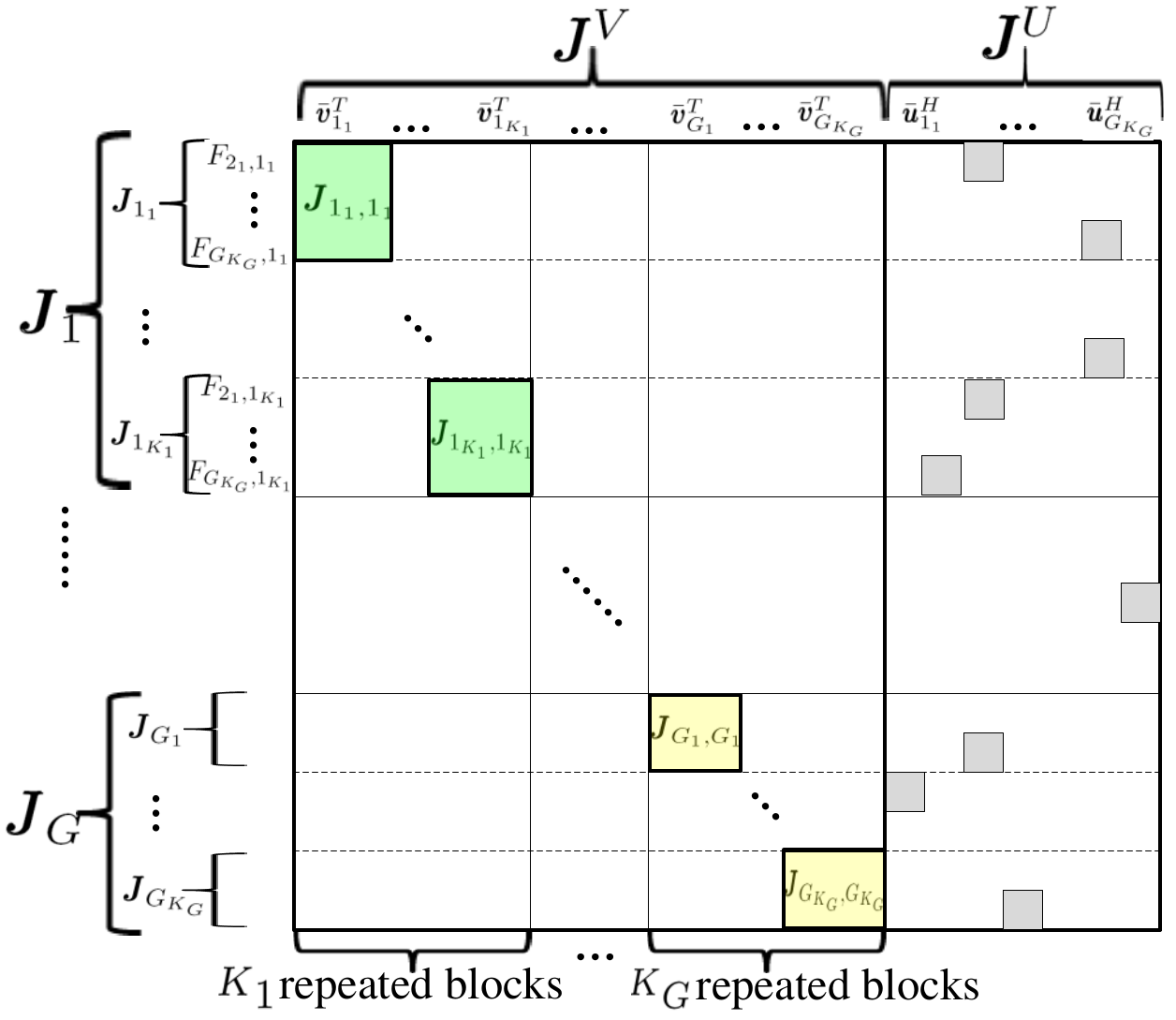}
\caption{Structure of $\pmb{J}(\bar{\pmb{H}}_0)$ of MIMO-IBC
$\prod_{i=1}^{G}(M_i\times \prod_{k=1}^{K_i}(N_{i_k},1))$.}
\label{fig:Structure_IBC}
\end{figure}

In the following, we construct an invertible Jacobian matrix with
such \emph{sparse structure} and \emph{repeated structure}.

Essentially, the existence of an invertible Jacobian matrix implies that all the ICIs in the corresponding network can be eliminated with linear IA, as analyzed with a bipartite graph in \cite{Luo2012}. This suggests that in order to construct an invertible Jacobian matrix we need to find a way to assign each of the variables in the transmit and receive vectors to each of the ICIs.

The structure of the Jacobian matrix in Fig. \ref{fig:Structure_IBC}
gives rise to the following observation: the transmit variable
assignment is not as flexible as the receive variable assignment.
Specifically, as shown from the proof of \emph{Lemma
\ref{Lemma:CanNot_Extended}}, the \emph{repeated structure} of
$\pmb{J}^{V}(\bar{\pmb{H}}_0)$ requires that if one transmit vector
of BS$_j$ is assigned to avoid the ICI to a user in other cell, the
other transmit vectors of BS$_j$ have also to avoid the ICI to the
same user. In other words, all the transmit vectors at one BS must
avoid generating ICIs to the same user in other cell.\footnote{It
means that the multiple ICIs between one BS and one user should be
eliminated either by using the spatial resources of the BS or by the
user. In fact, such a requirement can be satisfied only for the
system considered in \emph{Lemma \ref{Lemma:Repeated}} but not for
the system in \emph{Lemma \ref{Lemma:CanNot_Extended}}.} This leads
to the difficulty to construct an invertible Jacobian matrix for the
MIMO-IBC considered in \emph{Lemma \ref{Lemma:CanNot_Extended}},
where the BS can only avoid partial ICIs it generated but it does
not know which ICIs it should avoid. By contrast, the receive
variable assignment in a MIMO-IBC with $d = 1$ is flexible, because
$\pmb{J}^{U}(\bar{\pmb{H}}_0)$ does not have the \emph{repeated
structure}. By applying the result in \emph{Lemma
\ref{Lemma:Sparse}}, $\pmb{J}^{U}(\bar{\pmb{H}}_0)$ can be set as a
sub-matrix of a permutation matrix.

Inspired by this observation, we can first construct
$\pmb{J}^{U}(\bar{\pmb{H}}_0)$, i.e., assign the variables in the
receive vector to deal with some ICIs, using the way of perfect
matching. Then, we construct  $\pmb{J}^{V}(\bar{\pmb{H}}_0)$ to deal
with the remaining ICIs. To circumvent the confliction between
allowing the transmit vectors of each BS to avoid different ICIs and
ensuring the \emph{repeated structure} of the Jacobian matrix, we
only reserve enough variables in these transmit vectors but do not
assign variables to eliminate specific ICIs. Such an idea translates
to the following two rules to construct the invertible Jacobian
matrix.

\begin{itemize}
  \item \emph{Rule 1}: All the elements in $\pmb{J}^{U}(\bar{\pmb{H}}_0)$ are set as the corresponding elements in $\pmb{D}^{\mathcal{W}}$, where $\pmb{D}^{\mathcal{W}}$ is a permutation matrix obtained from \emph{Lemma \ref{Lemma:Sparse}}.
  \item \emph{Rule 2}: All the elements in $\pmb{J}_{j_1,j_1}^{V}(\bar{\pmb{H}}_0)$ are set to ensure that its arbitrary
  $M_j-K_j$ row vectors are linearly independent, and
  $\pmb{J}_{j_l,j_l}^{V}(\bar{\pmb{H}}_0)=\pmb{J}_{j_1,j_1}^{V}(\bar{\pmb{H}}_0),~l=2,\cdots,K_j$ that ensures the \emph{repeated structure}.
\end{itemize}

Now we prove that the constructed Jacobian matrix following these
rules is invertible. Since $\pmb{J}^{V}(\bar{\pmb{H}}_0)$ is a block
diagonal matrix, the nonzero blocks in different matrices of
$\pmb{J}_{j_l}^{V}(\bar{\pmb{H}}_0)$ are non-overlapping. Since the
elements in $\pmb{J}^{U}(\bar{\pmb{H}}_0)$ are set from the
permutation matrix $\pmb{D}^{\mathcal{W}}$, there is at most one
nonzero element in each column or row of
$\pmb{J}^{U}(\bar{\pmb{H}}_0)$. This indicates that the nonzero
elements in different matrices of
$\pmb{J}_{j_l}^{U}(\bar{\pmb{H}}_0)$ are also non-overlapping. As a
result, the nonzero elements in different blocks of
$\pmb{J}_{j_l}(\bar{\pmb{H}}_0) =
[\pmb{J}_{j_l}^{V}(\bar{\pmb{H}}_0),\pmb{J}_{j_l}^{U}(\bar{\pmb{H}}_0)]$
are non-overlapping. Considering the definition in
\eqref{Eq:Jacobian_Matrix_partitions}, we have
\begin{align}\label{Eq:Rank_Jacobian}
  \mathrm{rank}\left(\pmb{J}(\bar{\pmb{H}}_0)\right) = \sum_{j=1}^{G}\sum_{l=1}^{K_j} \mathrm{rank}\left(\pmb{J}_{j_l}(\bar{\pmb{H}}_0)\right)
\end{align}

In \emph{Rule 1}, the perfect matching ensures that there are
$\sum_{i=1}^{G}K_i-M_j$ ones in $\pmb{J}_{j_l}^{U}(\bar{\pmb{H}}_0)$
that are scattered in different rows, and then
$\mathrm{rank}(\pmb{J}_{j_l}^{U}(\bar{\pmb{H}}_0))=\sum_{i=1}^{G}K_i-M_j$.

Using elementary transformations, we can eliminate
$\sum_{i=1}^{G}K_i-M_j$ row vectors of
$\pmb{J}_{j_l}^{V}(\bar{\pmb{H}}_0)$ with nonzero elements and leave $M_j-K_j$ independent
row vectors in $\pmb{J}_{j_l}^{V}(\bar{\pmb{H}}_0)$. In this way,
the nonzero elements in $\pmb{J}_{j_l}^{U}(\bar{\pmb{H}}_0)$ and the transformed $\pmb{J}_{j_l}^{V}(\bar{\pmb{H}}_0)$ are located in different rows
of $\pmb{J}_{j_l}(\bar{\pmb{H}}_0)$. Therefore,
$\mathrm{rank}(\pmb{J}_{j_l}^{V}(\bar{\pmb{H}}_0))=M_j-K_j$ and
$\mathrm{rank}(\pmb{J}_{j_l}(\bar{\pmb{H}}_0))=
\mathrm{rank}(\pmb{J}_{j_l}^{V}(\bar{\pmb{H}}_0))+
\mathrm{rank}(\pmb{J}_{j_l}^{U}(\bar{\pmb{H}}_0))=\sum_{i=1,i\neq
j}^{G}K_i$. After substituting to \eqref{Eq:Rank_Jacobian}, we have
\begin{align}
\mathrm{rank}\left(\pmb{J}(\bar{\pmb{H}}_0)\right)=\sum_{j=1}^{G}\sum_{l=1}^{K_j}\sum_{i=1,i\neq
j}^{G}K_i=\sum_{j=1}^{G}\sum_{i=1,i\neq
j}^{G}K_jK_i
\end{align}
Comparing with \eqref{Eq:Rank_d1}, we know that $\pmb{J}(\bar{\pmb{H}}_0)$ is invertible.
Now, \emph{Theorem \ref{Theorem:NS_conditions}} is proved.
\end{proof}

\section{Discussion: Proper vs Feasible}\label{Sec:Discussion}
In this section, we discuss the connection between the proper and
feasibility conditions of the linear IA for MIMO-IBC by analyzing
and comparing \emph{Theorem \ref{Theorem:Necessary_Condition}} and
\emph{Theorem \ref{Theorem:NS_conditions}}. We also show the
relationship of our proved necessary and sufficient conditions with
existing results in the literature.

\subsection{``Proper''=``Feasible''}
For a class of MIMO-IBC with configuration
$\prod_{i=1}^{G}(M_i\times \prod_{k=1}^{K_i}(N_{i_k},d))$
where $M_i\geq K_id$ and $N_{i_k}\geq d$, from
\emph{Theorem \ref{Theorem:NS_conditions}} we know that when both
$M_i$ and $N_{i_k}$ are divisible by $d$, the MIMO-IBC is feasible
if it is proper. This immediately leads to the following conclusion:
\emph{when $d=1$, a proper MIMO-IBC is always feasible for arbitrary $M_i$
and $N_{i_k}$.}

When $d>1$, since there are too many cases of general MIMO-IBC to
describe and analyze, in the sequel we only focus on the symmetric
MIMO-IBC. We first show the ``proper condition'' for the symmetric
MIMO-IBC.

\begin{corollary}\label{Corollary:Proper}
For a symmetric MIMO-IBC with configuration $(M\times (N,d)^{K})^{G}$, the second
necessary condition in \emph{Theorem
\ref{Theorem:Necessary_Condition}}, i.e., the proper condition in
\eqref{Eq:Necessary_Condition_Interference}, reduces to
\begin{align}
\label{Eq:NS_Condition_Interference_Sym}
M+N\geq (GK+1)d
\end{align}
\end{corollary}
\begin{proof}
See Appendix \ref{App_Corollary1}.
\end{proof}

Note that \eqref{Eq:NS_Condition_Interference_Sym} was also obtained
in \cite{IBC_L_cell_LJD} from counting the total number of variables
and equations. However, it was not proposed as the proper condition.
From the definition of the proper system in \cite{Yetis2010}, a
system is proper iff for \emph{every} subset of the equations, the
number of the variables involved is at least as large as the number
of the equations. This means that to prove
\eqref{Eq:NS_Condition_Interference_Sym} as the proper condition, we
need to check: if \eqref{Eq:NS_Condition_Interference_Sym}
satisfies, whether \eqref{Eq:Necessary_Condition_Interference}
always holds for \emph{arbitrary} sets
$\mathcal{I}\subseteq\mathcal{J}$ and
$\mathcal{K}_i\subseteq\{1,\cdots,K\}$.

From \emph{Theorem \ref{Theorem:NS_conditions}} and \emph{Corollary
\ref{Corollary:Proper}} we know that for a symmetric MIMO-IBC with
$M\geq Kd$ and $N\geq d$, when $M$ and $N$ are divisible by $d$, the
IA of the symmetric MIMO-IBC will be feasible if the system is
proper. Next, we derive from \emph{Theorem
\ref{Theorem:NS_conditions}} that for a more general class of
symmetric MIMO-IBC, the IA will be feasible if the system is proper.
\begin{corollary}\label{Corollary:Feasible_sym}
For a symmetric MIMO-IBC, when
\begin{align}
\label{Eq:NS_Condition_Signal_Sym_x}
&M\geq (K+p)d,~N\geq ((G-1)K+1-p)d\nonumber\\
&\quad \quad \quad \exists p\in\{0,\cdots,(G-1)K\}
\end{align}
the IA is feasible.
\end{corollary}

\begin{proof}
See Appendix \ref{App_Corollary2}.
\end{proof}

This is the sufficient and necessary condition of the IA
feasibility, where $M$ and $N$ are not necessary to be divisible by
$d$ or $K$. Since for the symmetric MIMO-IBC the condition that
either $M$ or $N$ is divisible by $d$ is one special case of
\eqref{Eq:NS_Condition_Signal_Sym_x}, when either $M$ or $N$ is
divisible by $d$, the proper symmetric MIMO-IBC is feasible.

In literature, the sufficiency has been proved only for three
specific MIMO-IBC systems \cite{Suh2011,IBC_L_cell_Lovel,Shin2011}
and for two special classes of MIMO-IC \cite{Luo2012,Tse2011}.

For the three MIMO-IBC systems with $G=2$, $d=1$, $M =N =K+1$ in
\cite{Suh2011}, with $G=2$, $M =(K+1)d$, $N=Kd$ in \cite{Shin2011}
and with $G=2$, $M =Kd$, $N=(K+1)d$ in \cite{IBC_L_cell_Lovel}, the
sufficiency was proved implicitly by proposing closed-form linear IA
algorithms. We can see that these configurations satisfy
\eqref{Eq:NS_Condition_Interference_Sym} and
\eqref{Eq:NS_Condition_Signal_Sym_x}, i.e., the three systems are
proper and feasible, which are special cases of our results in
\emph{Corollary 2}.

For the two classes of MIMO-IC in \cite{Luo2012,Tse2011}, we can extend their results into MIMO-IBC by the following proposition.
\begin{proposition}\label{Proposition:Extend}
If there exists an
invertible Jacobian matrix for a class of MIMO-IC with configuration $\prod_{i=1}^{G}(M_i\times N_i,K_i)$, there will exist an invertible Jacobian matrix for a class of MIMO-IBC with configuration $\prod_{i=1}^{G}(M_i\times
\prod_{k=1}^{K_i}(N_{i_k},1))$.
\end{proposition}
\begin{proof}
See Appendix \ref{App_Lemma4}.
\end{proof}

According to \emph{Proposition \ref{Proposition:Extend}}, the sufficiency proof for the class of MIMO-IC in \cite{Luo2012} where either $M$ or $N-d$ is divisible by $d$ can be extended to a class of
MIMO-IBC where either $M$ or $N-d$ is divisible by $Kd$, which is a sub-class of those shown in \emph{Corollary
\ref{Corollary:Feasible_sym}}.
Similarly, the sufficiency proof for the class of MIMO-IC in \cite{Tse2011} where $G \geq 3$ and $M = N$ can be extended to a class of
MIMO-IBC where $G \geq 3$ and $M=N+(K-1)d$.
It is not hard to show that except for the cases when $(G-1)K$ is odd, all extended cases from \cite{Tse2011} are the special cases of those in \emph{Corollary
\ref{Corollary:Feasible_sym}}. When these extended MIMO-IBC systems satisfy the proper condition in \eqref{Eq:NS_Condition_Interference_Sym}, they are feasible.

\subsection{``Proper''$\neq$``Feasible''}

For a symmetric MIMO-IBC, the third necessary condition in \emph{Theorem
\ref{Theorem:Necessary_Condition}}, i.e., \eqref{Eq:Necessary_Condition_Compress}, reduces to
\begin{align}
\label{Eq:Necessary_Condition_Compress_Sym1}
\max\{pM,~qN\}\geq pKd + qd
\end{align}
where $p$ and $q$ were defined after \eqref{Eq:Constraint_ICI_free_Multiple}.

In a symmetric MIMO-IBC where $M\geq Kd$ and $N
\geq d$, when $M$ and $N$ satisfy
\eqref{Eq:NS_Condition_Interference_Sym} but do not satisfy
\eqref{Eq:Necessary_Condition_Compress_Sym1}, the MIMO-IBC is \emph{proper
but infeasible}. However, some conditions in \eqref{Eq:Necessary_Condition_Compress_Sym1} can be derived from \eqref{Eq:NS_Condition_Interference_Sym}. To investigate the proper but infeasible region of antenna configuration, we need to find which necessary conditions in \eqref{Eq:Necessary_Condition_Compress_Sym1} are not included in the proper condition.

{
\begin{table*}[htb!]\centering
\caption{Necessary conditions other than proper condition for
symmetric proper systems}\label{Table:Necessary_condition}
\begin{tabular}{l| l |c|c|c }
\hline
\multirow{2}{*}{ References} &  \multirow{2}{*}{Necessary conditions other than proper condition} & \multicolumn{3}{c}{Corresponding proper but infeasible cases} \\
\cline{3-5}
   &    &  Case I & Case II &  Other cases except I and II\\
\hline
\emph{Corollary \ref{Corollary:Proper_infeasible_sym}}  &  $\left\{\begin{array}{l}
              \max\{M,(G-1)KN\}\geq GKd \\
              \max\{(G-1)M,N\}\geq ((G-1)K+1)d
            \end{array}
 \right.,~\forall K,G$& $\forall K,G$ & $\forall K,G$ &  \\
\hline
\cite{Luo2012}$^{*}$,\cite{Tse2011}$^{*}$ &  $\max\{M,N\}\geq (K+1)d,~\forall K,G$ & & $G=2, \forall K$ & \\
\hline
\cite{IBC_Guillaud} &  $\left\{\begin{array}{l}
              \max\{M,(L-1)KN\}\geq LKd \\
              \max\{(L-1)M,KN\}\geq LKd
            \end{array}
 \right.,~L=2,\cdots,G, ~\forall K,G$& $\forall K,G$ & $K=1, \forall G$ &  \\
\hline
\cite{IBC_L_cell_Lovel} &  $\left\{\begin{array}{l}
              \max\{M,(G-1)N\}\geq (K+G-1)d \\
              \max\{(G-1)M,KN\}\geq (K+G-1)d
            \end{array}
 \right.,~\forall K,G$ & $K=1, \forall G$ & $K=1, \forall G$ & \\
\hline
\cite{Jafar_3cell},\cite{Tse_3cell_2011} &
$\left\{\begin{array}{l}
              \max\{LM,(L+1)N\}\geq (2L+1)d \\
              \max\{(L+1)M,LN\}\geq (2L+1)d
            \end{array}
 \right.,~\forall L\geq1, \left\{\begin{array}{l}
              K=1 \\
              G=3 \\
            \end{array}\right.
$& $K=1,G=3$ & $K=1,G=3$ & $K=1,G=3$\\
\hline
\cite{Jafar_IAchain_Kcell} &
$\left\{\begin{array}{l}
              \max\{M,(G-1)N\}\geq Gd \\
              \max\{(G-1)M,N\}\geq Gd \\
              \frac{(G-1)N}{G^2-G-1}\geq d, \forall \frac{G-1}{G(G-2)}\geq \frac{M}{N}\geq\frac{G}{G^2-G-1}\\
              \frac{(G-1)M}{G^2-G-1}\geq d, \forall \frac{G-1}{G(G-2)}\geq \frac{N}{M}\geq\frac{G}{G^2-G-1}
            \end{array}
 \right.,~\left\{\begin{array}{l}
              K=1 \\
              \forall G \\
            \end{array}\right.
            $ & $K=1,\forall G$ & $K=1,\forall G$ & $K=1,\forall G$ \\
\hline
\end{tabular}
\end{table*}
}

\begin{corollary}\label{Corollary:Proper_infeasible_sym}
For a symmetric MIMO-IBC where $M\geq Kd$ and $N \geq d$, there exist
at least two necessary conditions that are not included in the proper condition as follows
\begin{subequations}
\begin{align}\label{Eq:Proper_Infeasible_Condition1}
&\max\{M,(G-1)KN\}\geq GKd\\
\label{Eq:Proper_Infeasible_Condition2}
&\max\{(G-1)M,N\}\geq ((G-1)K+1)d
\end{align}
\end{subequations}
which lead to two proper but infeasible cases,
\begin{subequations}
\begin{align*}
\mathrm{Case~I:}~& \left\{\begin{array}{c}
           \max\{M,(G-1)KN\}<GKd \\
           M+N\geq (GK+1)d
         \end{array}
  \right. \\
\mathrm{Case~II:}~&\left\{\begin{array}{c}
           \max\{(G-1)M,N\}<((G-1)K+1)d \\
           M+N\geq (GK+1)d
         \end{array}
  \right.
\end{align*}
\end{subequations}
When $G=2$ and $K=1$, Case I is the same as Case II, otherwise these two cases are different.

\end{corollary}

\begin{proof}
See Appendix \ref{App_Corollary3}.
\end{proof}

Many necessary conditions other than the proper condition were provided for various MIMO-IC \cite{Jafar_IAchain_Kcell,Luo2012,Tse2011,Jafar_3cell,Tse_3cell_2011,Win_IC} and MIMO-IBC
\cite{IBC_L_cell_Lovel,IBC_Guillaud}. In \cite{Luo2012,Tse2011}, a necessary condition of $\max\{M,N\}\geq 2d$ was provided for symmetric MIMO-IC, which is not difficult to be extended into symmetric MIMO-IBC as
$\max\{M,N\}\geq (K+1)d$. In \cite{Jafar_IAchain_Kcell,Jafar_3cell,Tse_3cell_2011}, the methods to derive the necessary conditions are only applicable for MIMO-IC.
Since some of the necessary conditions can be derived from the proper condition, we only compare the corresponding proper but infeasible cases, which can be obtained from the necessary conditions after some regular but
tedious derivations. For conciseness, we omit the details of the
derivation. In fact, no more than two conditions in \cite{IBC_L_cell_Lovel,Luo2012,Tse2011,IBC_Guillaud} cannot be derived from the proper condition, which leads to no more than two
proper but infeasible cases.

We list the existing and our extended results in Table
\ref{Table:Necessary_condition}. It is shown that all existing
proper but infeasible cases except those in
\cite{Jafar_IAchain_Kcell,Jafar_3cell,Tse_3cell_2011} are special
cases of \emph{Corollary \ref{Corollary:Proper_infeasible_sym}}.

For the symmetric MIMO-IC, all the necessary conditions in
\cite{Jafar_IAchain_Kcell,Jafar_3cell,Tse_3cell_2011} cannot be
derived from the proper condition. For the symmetric three-cell
MIMO-IC in \cite{Jafar_3cell,Tse_3cell_2011}, when $L=1$, the
obtained proper but infeasible cases are included in Cases I and II
of our results, when $L>1$, their necessary conditions lead to other
proper but infeasible cases, where $L$ is an arbitrary positive
integer. For the symmetric $G-$cell MIMO-IC in
\cite{Jafar_IAchain_Kcell}, there are four necessary conditions that
correspond to four proper but infeasible cases. Two of the cases are
included in Cases I and II of our results, but another two cases are
not. Consequently, when $K=1$, the necessary conditions in
\cite{Jafar_IAchain_Kcell,Jafar_3cell,Tse_3cell_2011} are more
general than ours, but the results are only applicable for MIMO-IC.

\subsection{An example} In Fig. \ref{fig:Proper_Infeasible_IBC_G4K3}, we
illustrate the feasible and infeasible regions (i.e., the
corresponding system configuration) with an example. The
feasible results from \emph{Corollary \ref{Corollary:Feasible_sym}}
are shown by horizontal lines. The extended results from
\cite{Luo2012} and from \cite{Tse2011} through \emph{Proposition \ref{Proposition:Extend}} are respectively in dash
and dot-dash
lines. It is shown that all the extended results from \cite{Luo2012} and from \cite{Tse2011} are
special cases of those from \emph{Corollary
\ref{Corollary:Feasible_sym}} since $(G-1)K$ is even in the example.

The two proper but infeasible cases in \emph{Corollary
\ref{Corollary:Proper_infeasible_sym}} are highlighted with red background. In the proper region, except for the region that has been
proved to be feasible in \emph{Corollary
\ref{Corollary:Feasible_sym}} and
that has been proved to be infeasible in
\emph{Corollary
\ref{Corollary:Proper_infeasible_sym}}, the feasibility of the
remaining region is still unknown.

\begin{figure}[htb!]
\centering
\includegraphics[width=1.0\linewidth]{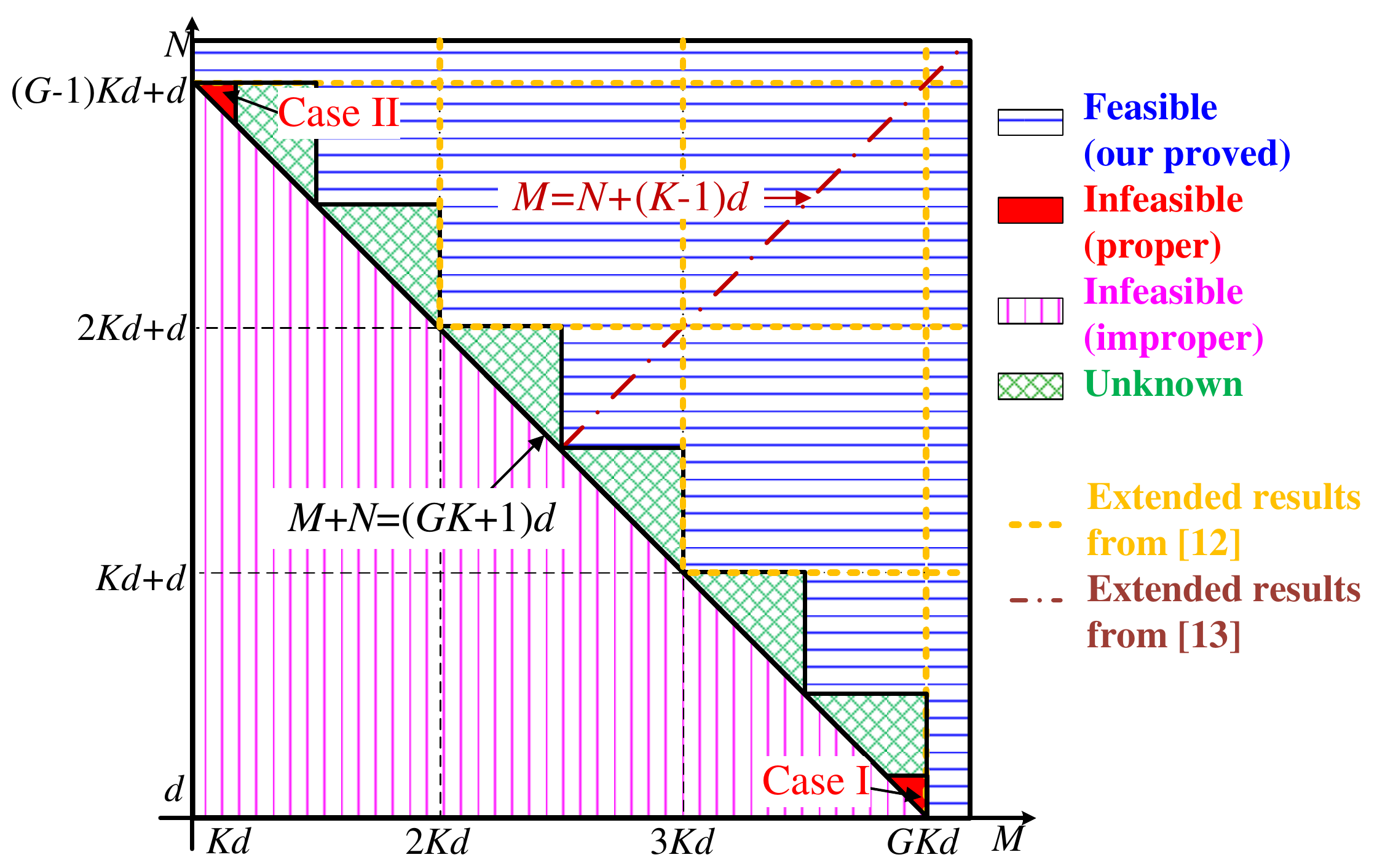}
\caption{Feasible and infeasible regions of linear IA for the
MIMO-IBC supporting overall $GKd$ data streams, $G=4,K=2$, where
$(G-1)K$ is even.} \label{fig:Proper_Infeasible_IBC_G4K3}
\end{figure}

\section{Conclusion}\label{Sec:Conclusion}

In this paper, we proposed and proved necessary conditions of linear
IA feasibility for general MIMO-IBC with constant coefficients.
A necessary condition other than  proper condition was posed
to ensure the elimination of a kind of irreducible interference. The
existence conditions of the reducible and irreducible interference
were provided, which depend on the difference in spatial dimension
between a base station and multiple users or between a user and
multiple base stations.

We proved necessary and sufficient conditions for a special class of
MIMO-IBC by finding an invertible Jacobian matrix, which include
existing results in literature as special cases. Our analysis showed
that when multiple ICIs between one BS and one user can be
eliminated either by the BS or by the user, there exists an
invertible Jacobian matrix that is a permutation matrix. By
contrast, when these ICIs must be eliminated by sharing spatial
resources between the BS and the user, the Jacobian matrix cannot be
set as a permutation matrix owing to its \emph{repeated structure}.
To deal with the conflicting requirements on the \emph{sparse
structure} and the \emph{repeated structure} of MIMO-IBC, a general
rule to construct an invertible Jacobian matrix was proposed, by
exploiting the flexibility of MIMO-IBC in assigning the spatial
sources at the users. Finally, we analyzed the feasible, proper but
infeasible, and unknown regions in antenna configuration for a
proper symmetric MIMO-IBC. The analysis is not applicable to the
MIMO-IBC with symbol extension.

\appendices
\numberwithin{equation}{section}
\section{Proof of Lemma \ref{Lemma:Sparse}}\label{App_Lemma3}
\renewcommand{\theequation}{A.\arabic{equation}}
\begin{proof}
The relationship between the equations and variables in \eqref{Eq:Constraint_ICI_free_New5} can be represented by a bipartite graph, denoted by $C=\left(\mathcal{X},\mathcal{Y},\mathcal{E}\right)$, where $\mathcal{Y}$ is the set of vertices representing the scalar equations, $\mathcal{X}$ is the set of vertices representing the scalar variables, $\mathcal{E}$ is the
set of edges and $[Y_m,X_n]\in \mathcal{E}$ iff equation $Y_m$
contains variable $X_n$, where $X_m$ and $Y_m$ are the $m$th elements
in $\mathcal{X}$ and $\mathcal{Y}$, respectively.

Hall's
theorem \cite[Theorem 3.1.11]{Grahp_Theory} indicates that in a bipartite graph, a perfect matching exists iff $|N(\mathcal{S})| \geq |\mathcal{S}|,~\forall \mathcal{S}\subseteq \mathcal{Y}$, where
$N(\mathcal{S})$ is the set of all vertices adjacent to some
elements of $\mathcal{S}$. In a general MIMO-IBC with configuration $\prod_{i=1}^{G}(M_i\times
\prod_{k=1}^{K_i}(N_{i_k},d_{i_k}))$, $|\mathcal{S}|=\sum_{(i,j)\in\mathcal{I}}d_j\sum_{k\in
\mathcal{K}_i}d_{i_k}$ and
$|N(\mathcal{S})|=\sum_{j:(i,j)\in \mathcal{I}}
\left(M_j-d_j\right)d_j+ \sum_{i:(i,j)\in
\mathcal{I}}\sum_{k\in
\mathcal{K}_i}\left(N_{i_k}-d_{i_k}\right)d_{i_k}$, therefore,
$|N(\mathcal{S})| \geq |\mathcal{S}|$ is actually the proper
condition for the MIMO-IBC. Consequently, according to
Hall's theorem we know that when the MIMO-IBC is proper and $L_v=L_e$,
there exists a perfect matching in the bipartite graph.

Denote a perfect matching in $\mathcal{C}$ as $\mathcal{W}$, it is clear that $\mathcal{W}\subseteq \mathcal{E}$.
The perfect matching $\mathcal{W}$ can be represented by an adjacency matrix  $\pmb{D}^{\mathcal{W}}$, that represents which vertices in one set of a graph are connected to the vertices in the other set.
Then, the $(m,n)$th element of $\pmb{D}^{\mathcal{W}}$ is
\begin{align}\label{Eq:Pefect_Matching}
  D_{m,n}^{\mathcal{W}}=\left\{\begin{array}{ll}
                   1,&~\forall [Y_m,X_n]\in \mathcal{W} \\
                   0,&~\forall [Y_m,X_n]\notin \mathcal{W}
                 \end{array}
  \right.
\end{align}
In the perfect matching, the two sets of vertices have a one-to-one
mapping relationship. Therefore, $\pmb{D}^{\mathcal{W}}$ is a
permutation matrix such that the elements of $\pmb{D}^{\mathcal{W}}$
satisfy $\sum_{m=1}^{L_e}D_{m,n}^{\mathcal{W}}=1$ and
$\sum_{n=1}^{L_e}D_{m,n}^{\mathcal{W}}=1$.

From \eqref{Eq:Constraint_ICI_free_SingleStream}, we know that the sparse property of Jacobian matrix indicates that the $(m,n)$th element satisfies
\begin{align}\label{Eq:jacobi}
  J_{m,n}=\frac{\partial Y_m}{\partial X_n}= 0,&~\forall [Y_m,X_n]\notin \mathcal{E}
\end{align}
Since $\mathcal{W} \subseteq \mathcal{E}$, from
\eqref{Eq:Pefect_Matching} we have $D_{m,n}^{\mathcal{W}}=0, \forall
[Y_m,X_n]\notin \mathcal{E}$. Compared with \eqref{Eq:jacobi}, we
know that $\pmb{D}^{\mathcal{W}}$ has the same \emph{sparse
structure} as the Jacobian matrix for a proper system with
$L_v=L_e$.
\end{proof}

\section{Proof of Lemma \ref{Lemma:Repeated}}\label{App_Lemma1}
\numberwithin{equation}{section}
\renewcommand{\theequation}{B.\arabic{equation}}
\begin{proof}
When $d_{i_k}=d$, and $M_j$ and $N_{i_k}$ are
divisible by $d$, \eqref{Eq:Constraint_ICI_free_New5} can be
further rewritten as
\begin{align}\label{Eq:Constraint_ICI_free_New5d}
&\pmb{F}_{i_k,j_l}(\bar{\pmb{H}}_{0})
=\bar{\pmb{H}}_{0~i_k,j}^{(2)}\bar{\pmb{V}}_{j_l} +\bar{\pmb{U}}_{i_k}^{H}\bar{\pmb{H}}_{0~i_k,j_l}^{(3)}\\
=&\sum_{t=1}^{M_j/d-K_j}\bar{\pmb{H}}_{0~i_k,j}^{(2),t}\bar{\pmb{V}}_{j_l(t)}
+\sum_{s=1}^{N_{i_k}/d-1}\bar{\pmb{U}}^{H}_{i_k(s)}\bar{\pmb{H}}_{0~i_k,j_l}^{(3),s}
=\pmb{0},~\forall i\neq j \nonumber
\end{align}
where $\bar{\pmb{V}}_{j_l(t)}$ and $\bar{\pmb{H}}_{0~i_k,j}^{(2),t}$
are the $t$th block of size $d\times d$ in $\bar{\pmb{V}}_{j_l}$ and
$\bar{\pmb{H}}_{0~i_k,j}^{(2)}$, $\bar{\pmb{U}}_{i_k(s)}$ and
$\bar{\pmb{H}}_{0~i_k,j_l}^{(3),s}$ are the $s$th block of size
$d\times d$ in $\bar{\pmb{U}}_{i_k}$ and
$\bar{\pmb{H}}_{0~i_k,j_l}^{(3)}$.

Then, from \eqref{Eq:Jacobian_Matrix_IBC_V} and \eqref{Eq:Jacobian_Matrix_IBC_U}, the elements of $\pmb{J}(\bar{\pmb{H}}_{0})$ become
\begin{subequations}
\begin{align}\label{Eq:Jacobian_Matrix_IBC_Vd}
\frac{\partial\mathrm{vec}\{\pmb{F}_{i_k,j_l}(\bar{\pmb{H}}_{0})\}}
  {\partial\mathrm{vec}\{\bar{\pmb{V}}_{m_n(t)}\}}&=
\left\{
\begin{array}{ll}
\bar{\pmb{H}}_{0~i_k,j}^{(2),t}\otimes\pmb{I}_{d},&~\forall m_n=j_l\\
\pmb{0}_{d^2},&~\forall m_n\neq j_l
\end{array}
  \right.
  \\
\label{Eq:Jacobian_Matrix_IBC_Ud}
  \frac{\partial\mathrm{vec}\{\pmb{F}_{i_k,j_l}(\bar{\pmb{H}}_{0})\}}
  {\partial\mathrm{vec}\{\bar{\pmb{U}}^{H}_{m_n(s)}\}}&=
\left\{
\begin{array}{ll}
\pmb{I}_{d} \otimes (\bar{\pmb{H}}_{0~i_k,j_l}^{(3),s})^{T},  &~\forall m_n=i_k\\
\pmb{0}_{d^2},&~\forall  m_n\neq i_k
\end{array}
  \right.
\end{align}
\end{subequations}
where $t=1,\cdots,M_j/d-K_j$, $s=1,\cdots,N_{i_k}/d-1$.

When let
$\bar{\pmb{H}}_{0~i_k,j}^{(2),t}=\bar{h}_{0~i_k,j}^{(2),t}\pmb{I}_{d}$ and
$\bar{\pmb{H}}_{0~i_k,j_l}^{(3),s}=\bar{h}_{0~i_k,j_l}^{(3),s}\pmb{I}_{d}$,
where $\bar{h}_{0~i_k,j}^{(2),t}$ and $\bar{h}_{0~i_k,j_l}^{(3),s}$ are the
$(1,1)$th elements of $\bar{\pmb{H}}_{0~i_k,j}^{(2),t}s$ and
$\bar{\pmb{H}}_{0~i_k,j_l}^{(3),s}$, respectively, we have
$\bar{\pmb{H}}_{0~i_k,j}^{(2),t}\otimes\pmb{I}_{d}=\bar{h}_{0~i_k,j}^{(2),t}\pmb{I}_{d^2}$
and $\pmb{I}_{d} \otimes
\bar{\pmb{H}}_{0~i_{K_i},j_l}^{(3),s}=\bar{h}_{0~i_k,j_l}^{(3),s}\pmb{I}_{d^2}$.
As a result, the Jacobian matrix for a MIMO-IBC with configuration $\prod_{i=1}^{G}(M_i\times
\prod_{k=1}^{K_i}(N_{i_k},d))$ can be rewritten as
$\pmb{J}(\bar{\pmb{H}}_{0})=\tilde{\pmb{J}}(\bar{\pmb{H}}_{0})\otimes \pmb{I}_{d^2}$, where
$\tilde{\pmb{J}}(\bar{\pmb{H}}_{0})$ has the same pattern of nonzero elements as the
Jacobian matrix for a MIMO-IBC with configuration $\prod_{i=1}^{G}(M_i/d \times
\prod_{k=1}^{K_i}(N_{i_k}/d,1))$. Therefore, once an invertible matrix $\tilde{\pmb{J}}(\bar{\pmb{H}}_{0})$ is obtained,
an invertible matrix $\pmb{J}(\bar{\pmb{H}}_{0})$ is obtained immediately. Moreover, if $\tilde{\pmb{J}}(\bar{\pmb{H}}_{0})$ is a permutation matrix, $\pmb{J}(\bar{\pmb{H}}_{0})$ is also a permutation matrix.
\end{proof}

\section{Proof of Lemma \ref{Lemma:CanNot_Extended}}\label{App_Lemma5}
\begin{proof}
In a general MIMO-IBC with configuration $\prod_{i=1}^{G}(M_i\times
\prod_{k=1}^{K_i}(N_{i_k},d_{i_k}))$, for an arbitrary data stream
of MS$_{i_k}$, the number of ICIs it experienced is an element in
$\phi_{i}$, and the number of variables in its effective receive
vector is $N_{i_k}-d_{i_k}$. When $N_{i_k}-d_{i_k} \leq
\sum_{j=1,j\neq i}^{G}d_j$ and $N_{i_k}-d_{i_k} \notin \phi_i$,
there will exist one BS (say BS$_j$) where the number of variables
in the effective receive vector of MS$_{i_k}$ is not large enough to
cancel all the $d_j$ ICIs generated from BS$_j$, denoted by
${Y}_{1},\cdots,{Y}_{d_j}$. When $N_{i_k}-d_{i_k}>0$, the effective
receive vector of MS$_{i_k}$ is able to cancel a part of the ICIs
from BS$_j$, which means that BS$_{j}$ cannot avoid all the ICIs to
the data stream of MS$_{i_k}$ considering $L_v=L_e$. Consequently,
these conditions imply that the $d_j$ ICIs from BS$_j$ to the data
stream of  MS$_{i_k}$ need to be jointly eliminated by BS$_j$ and
MS$_{i_k}$, rather than solely by BS$_j$ or MS$_{i_k}$.

We first show the structure of a Jacobian matrix if it is set as a permutation matrix $\pmb{D}^{\mathcal{W}}$.
Denote the $t$th variable of the $m$th transmit vector of BS$_j$ as $X_{m(t)}$, where $t=1,\cdots,M_j-d_j$ and $m=1,\cdots,d_j$. If an effective transmit variable $X_{m(t)}$ is assigned to avoid the ICI $Y_{m}$ in the perfect matching, from \eqref{Eq:Pefect_Matching}, we know that
setting $\pmb{J}(\bar{\pmb{H}}_0)=\pmb{D}^{\mathcal{W}}$ requires
$\partial Y_{m}/\partial X_{m(t)}=1$, otherwise $\partial Y_{m}/\partial X_{m(t)}=0$. Since ${Y}_{1},\cdots,{Y}_{d_j}$ need to be eliminated by BS$_j$ and MS$_{i_k}$ jointly, a perfect matching (that corresponds a permutation matrix that satisfies the \emph{sparse structure}) requires that
some of $\partial Y_{1}/\partial X_{1(t)},\cdots,\partial Y_{d_j}/\partial X_{d_j(t)}$ are ones and others are zeros.

According to the \emph{repeated structure} of the Jacobian matrix
shown in \eqref{Eq:Jacobian_Matrix_IBC_V} and
\eqref{Eq:Jacobian_Matrix_IBC_VV}, we have $\partial Y_{1}/\partial
X_{1(t)}=\cdots=\partial Y_{d_j}/\partial X_{d_j(t)}$. As a result,
any permutation matrix that satisfies the \emph{sparse structure} of
the Jacobian matrix cannot satisfy its \emph{repeated structure}.
Therefore, there does not exist a Jacobian matrix that is a
permutation matrix.
\end{proof}

\section{Proof of Corollary \ref{Corollary:Proper}}\label{App_Corollary1}
\numberwithin{equation}{section}
\renewcommand{\theequation}{D.\arabic{equation}}

\begin{proof}
For a symmetric MIMO-IBC, \eqref{Eq:Necessary_Condition_Interference}
becomes
\begin{align}
\label{Eq:IA_Condition_333}
\left(M-Kd\right)K_{\mathrm{T}}+\left(N-d\right)K_{\mathrm{R}}\geq
\sum_{(i,j)\in \mathcal{I}}K|\mathcal{K}_{i}|d
\end{align}
where $K_{\mathrm{T}}=\sum_{j\in \mathcal{I}_{\mathrm{T}}}K$,
$K_{\mathrm{R}}=\sum_{i\in
\mathcal{I}_{\mathrm{R}}}|\mathcal{K}_{i}|$,
$\mathcal{I}_{\mathrm{T}}=\{j|(i,j)\in \mathcal{I}\}$ and
$\mathcal{I}_{\mathrm{R}}=\{i|(i,j)\in \mathcal{I}\}$.
$\mathcal{I}_{\mathrm{T}}$ and $\mathcal{I}_{\mathrm{R}}$ denote the
index sets of the cells  in $\mathcal{I}$ that generate ICI and
suffer from the ICI, $K_{\mathrm{T}}$ and $K_{\mathrm{R}}$ are the
total numbers of users in the cells with indices in
$\mathcal{I}_{\mathrm{T}}$ and $\mathcal{I}_{\mathrm{R}}$,
respectively.

Define $\tilde{\mathcal{I}}\triangleq\{(i,j)|i\neq j,~\forall
j\in\mathcal{I}_{\mathrm{T}},~i\in\mathcal{I}_{\mathrm{R}}\}$, it is
easy to know $\mathcal{I}\subseteq
\tilde{\mathcal{I}}$.\footnote{For example, when
$\mathcal{I}=\{(1,3),(2,4)\}$, we have
$\mathcal{I}_{\mathrm{R}}=\{1,2\}$ and
$\mathcal{I}_{\mathrm{T}}=\{3,4\}$. From the definition of
$\tilde{\mathcal{I}}$, we know
$\tilde{\mathcal{I}}=\{(1,3),(1,4),(2,3),(2,4)\}$. Obviously,
$\mathcal{I}\subseteq \tilde{\mathcal{I}}$.} Therefore, the
right-hand side of \eqref{Eq:IA_Condition_333} satisfies
\begin{align}\label{Eq:Cardinality_I}
&\sum_{(i,j)\in \mathcal{I}}K|\mathcal{K}_{i}|d\leq \sum_{(i,j)\in \tilde{\mathcal{I}}}K|\mathcal{K}_{i}|d \\
=&
\sum_{j\in \mathcal{I}_{\mathrm{T}}}K\sum_{i\in\mathcal{I}_{\mathrm{R}},i\neq j}
|\mathcal{K}_{i}|=\sum_{i\in\mathcal{I}_{\mathrm{R}}}
|\mathcal{K}_{i}|\sum_{j\in \mathcal{I}_{\mathrm{T}},j\neq i}K \nonumber
\end{align}

Since $\mathcal{I}_{\mathrm{T}}\subseteq \{1,\ldots,G\}$, we have
$\sum_{i\in\mathcal{I}_{\mathrm{R}}} |\mathcal{K}_{i}|\sum_{j\in
\mathcal{I}_{\mathrm{T}},j\neq i}K\leq
\sum_{i\in\mathcal{I}_{\mathrm{R}}}|\mathcal{K}_{i}| \sum_{j=1,i\neq
j}^{G}K =(G-1)KK_{\mathrm{R}}$. Since
$\mathcal{I}_{\mathrm{R}}\subseteq \{1,\ldots,G\}$ and
$|\mathcal{K}_{i}|\leq K$, we have $\sum_{j\in
\mathcal{I}_{\mathrm{T}}}K\sum_{i\in\mathcal{I}_{\mathrm{R}},i\neq
j}|\mathcal{K}_{i}|\leq \sum_{j\in
\mathcal{I}_{\mathrm{T}}}K\sum_{i=1,i\neq j}^{G}K
=(G-1)KK_{\mathrm{T}} $. After substituting into
\eqref{Eq:Cardinality_I}, we obtain an upper-bound of the right-hand
side of \eqref{Eq:IA_Condition_333} as
\begin{align}\label{Eq:Cardinality_I3}
\sum_{(i,j)\in \mathcal{I}}K|\mathcal{K}_{i}|d\leq (G-1)Kd \min\{K_{\mathrm{R}},K_{\mathrm{T}}\}
\end{align}

Because $K_{\mathrm{T}} \geq \min\{K_{\mathrm{R}},K_{\mathrm{T}}\}$
and $K_{\mathrm{R}} \geq \min\{K_{\mathrm{R}},K_{\mathrm{T}}\}$, the
left-hand side of \eqref{Eq:IA_Condition_333} satisfies
\begin{align}\label{Eq:IA_Equality0}
&\left(M-Kd\right)K_{\mathrm{T}}+
\left(N-d\right)K_{\mathrm{R}}\nonumber\\
\geq &\left(M+N-(K+1)d\right)\min\{K_{\mathrm{R}},K_{\mathrm{T}}\}
\end{align}
From \eqref{Eq:NS_Condition_Interference_Sym}, we have
$M+N-(K+1)d\geq (G-1)Kd$. Substituting this inequity into
\eqref{Eq:IA_Equality0}, we obtain a lower-bound of the left-hand
side of \eqref{Eq:IA_Condition_333} as
\begin{align}\label{Eq:IA_Equality1}
\left(M-Kd\right)K_{\mathrm{T}}+
\left(N-d\right)K_{\mathrm{R}}\geq (G-1)Kd\min\{K_{\mathrm{R}},K_{\mathrm{T}}\}
\end{align}
Consider \eqref{Eq:Cardinality_I3} and \eqref{Eq:IA_Equality1}, we
obtain \eqref{Eq:IA_Condition_333}.
\end{proof}

\section{Proof of Corollary \ref{Corollary:Feasible_sym}}\label{App_Corollary2}
\begin{proof}
For notation simplicity, we define $M_p=(K+p)d$ and
$N_p=((G-1)K+1-p)d$ here. To prove the MIMO-IBC with configuration
$(M\times (N,d)^{K})^{G}$ where $M$ and $N$ satisfy $M\geq
M_p,~N\geq N_p,~\exists p$ to be feasible, we first prove its IA to
be feasible when $M=M_p,~N= N_p,~\exists p$.

Since $M_p+N_p=(GK+1)d,~\forall p\in\{0,\cdots,(G-1)K\}$, according
to \emph{Corollary \ref{Corollary:Proper}}, we know that the
MIMO-IBC with configuration $(M_p\times (N_p,d)^{K})^{G}$ is proper.
Because $M_p\geq Kd,~N_p\geq d$ and both $M_p$ and $N_p$ are
divisible by $d$, according to \emph{Theorem
\ref{Theorem:NS_conditions}}, the IA for the MIMO-IBC with
configuration $(M_p\times (N_p,d)^{K})^{G}$ is feasible.
\end{proof}

\section{Proof of Proposition \ref{Proposition:Extend}}\label{App_Lemma4}
\numberwithin{equation}{section}
\renewcommand{\theequation}{F.\arabic{equation}}
\begin{proof}
For conciseness, we use $\pmb{J}_{\mathrm{IBC}}(\bar{\pmb{H}}_0)$
and $\pmb{J}_{\mathrm{IC}}(\bar{\pmb{H}}_0)$ to denote the Jacobian
matrix of the MIMO-IBC with configuration $\prod_{i=1}^{G}(M_i\times
\prod_{k=1}^{K_i}(N_{i_k},1))$ and that of the MIMO-IC with
configuration $\prod_{i=1}^{G}(M_i\times N_i,K_i)$, respectively. To
show how to obtain an invertible
$\pmb{J}_{\mathrm{IBC}}(\bar{\pmb{H}}_0)$ from an invertible
Jacobian matrix $\pmb{J}_{\mathrm{IC}}(\bar{\pmb{H}}_0)$, we first
show the structure of $\pmb{J}_{\mathrm{IC}}(\bar{\pmb{H}}_0)$.

In the MIMO-IC, \eqref{Eq:Constraint_ICI_free_Linear} can be rewritten as
\begin{align}\label{Eq:Constraint_ICI_free_New_IC}
&{F}_{i_k,j_l}^{\mathrm{IC}}(\bar{\pmb{H}}_{0})
=\bar{\pmb{h}}_{0~i_k,j}^{(2)}\bar{\pmb{v}}_{j_l}+ \bar{\pmb{u}}_{i_k}^{H}\bar{\pmb{h}}_{0~i,j_l}^{(3)}=0,~\forall i\neq j
\end{align}
where ${F}_{i_k,j_l}(\cdot)$
represents the ICI from the
$l$th data stream transmitted from BS$_j$ to the $k$th data
stream received at MS$_i$.

From \eqref{Eq:Jacobian_Matrix_IBC_V},
\eqref{Eq:Jacobian_Matrix_IBC_U} and
\eqref{Eq:Constraint_ICI_free_New_IC}, we can obtain the elements of
the Jacobian matrix as follows,
\begin{subequations}
\begin{align}\label{Eq:Jacobian_Matrix_IC_V}
\frac{\partial{F}_{i_k,j_l}^{\mathrm{IC}}(\bar{\pmb{H}}_{0})}
{\partial\bar{\pmb{v}}_{m_n}}&=
\left\{
\begin{array}{ll}
\bar{\pmb{h}}_{0~i_k,j}^{(2)},&~\forall m_n=j_l\\
\pmb{0}_{1\times (M_m-K_m)}, &~\forall m_n\neq j_l
\end{array}
  \right.
  \\
\label{Eq:Jacobian_Matrix_IC_U}
  \frac{\partial{F}_{i_k,j_l}^{\mathrm{IC}}(\bar{\pmb{H}}_{0})}
  {\partial\bar{\pmb{u}}^{H}_{m_n}}&=
\left\{
\begin{array}{ll}
(\bar{\pmb{h}}_{0~i,j_l}^{(3)})^{T},  &~\forall m_n=i_k\\
\pmb{0}_{1\times (N_m-K_m)}, &~\forall m_n\neq i_k
\end{array}
  \right.
\end{align}
\end{subequations}

In \eqref{Eq:Jacobian_Matrix_IC_V} and \eqref{Eq:Jacobian_Matrix_IC_U},
one can see that the nonzero elements satisfy
\begin{subequations}
\begin{align}
\label{Eq:Jacobian_Matrix_Repetition_V}
&\frac{\partial F_{i_k,j_1}^{\mathrm{IC}}(\bar{\pmb{H}}_{0})}
{\partial\bar{\pmb{v}}_{j_1}}=\cdots =\frac{\partial{F}_{i_k,j_{K_j}}^{\mathrm{IC}}(\bar{\pmb{H}}_{0})}
{\partial\bar{\pmb{v}}_{j_{K_j}}}=\bar{\pmb{h}}_{0~i_k,j}^{(2)}\\
\label{Eq:Jacobian_Matrix_Repetition_U}
&\frac{\partial{F}_{i_1,j_l}^{\mathrm{IC}}(\bar{\pmb{H}}_{0})}
{\partial\bar{\pmb{u}}^{H}_{i_1}}=\cdots =\frac{\partial{F}_{i_{K_i},j_l}^{\mathrm{IC}}(\bar{\pmb{H}}_{0})}
{\partial\bar{\pmb{u}}^{H}_{i_{K_i}}}=\bar{\pmb{h}}_{0~i,j_l}^{(3)}
\end{align}
\end{subequations}

Comparing \eqref{Eq:Jacobian_Matrix_IBC_D1_V} and
\eqref{Eq:Jacobian_Matrix_IBC_D1_U} with
\eqref{Eq:Jacobian_Matrix_IC_V} and \eqref{Eq:Jacobian_Matrix_IC_U},
it is easy to find that when $N_{i_k}-1=N_i-K_i,~\forall i,k$, $\pmb{J}_{\mathrm{IBC}}(\bar{\pmb{H}}_0)$ has the same nonzero element pattern
with $\pmb{J}_{\mathrm{IC}}(\bar{\pmb{H}}_0)$.
Moreover, comparing \eqref{Eq:Jacobian_Matrix_Repetition_V1} and
\eqref{Eq:Jacobian_Matrix_Repetition_V}, we can see that the
repeated nonzero elements in
$\pmb{J}^{V}_{\mathrm{IBC}}(\bar{\pmb{H}}_0)$ have the same pattern
as those in $\pmb{J}^{V}_{\mathrm{IC}}(\bar{\pmb{H}}_0)$. By
contrast, comparing \eqref{Eq:Jacobian_Matrix_Repetition_U1} and
\eqref{Eq:Jacobian_Matrix_Repetition_U}, we can see that the nonzero
elements of $\pmb{J}^{U}_{\mathrm{IBC}}(\bar{\pmb{H}}_0)$ are
generic but those of $\pmb{J}^{U}_{\mathrm{IC}}(\bar{\pmb{H}}_0)$
are not since they are repeated. This suggests that the elements of
$\pmb{J}_{\mathrm{IBC}}(\bar{\pmb{H}}_0)$ are more flexible to be
set into any value than that of
$\pmb{J}_{\mathrm{IC}}(\bar{\pmb{H}}_0)$. Hence, if there exists an
invertible $\pmb{J}_{\mathrm{IC}}(\bar{\pmb{H}}_0)$, we can obtain
an invertible $\pmb{J}_{\mathrm{IBC}}(\bar{\pmb{H}}_0)$ by setting
$\pmb{J}_{\mathrm{IBC}}(\bar{\pmb{H}}_0)=\pmb{J}_{\mathrm{IC}}(\bar{\pmb{H}}_0)$.
\end{proof}

\section{Proof of Corollary \ref{Corollary:Proper_infeasible_sym}}\label{App_Corollary3}
\numberwithin{equation}{section}
\renewcommand{\theequation}{G.\arabic{equation}}
\begin{proof}
If $M$ and $N$ do not satisfy
\eqref{Eq:Necessary_Condition_Compress_Sym1}, we have
$\max\{pM,~qN\}< pKd + qd$, i.e., $pM < pKd + qd$ and $qN < pKd +
qd$. Considering $M+N\geq (GK+1)d$ in
\eqref{Eq:NS_Condition_Interference_Sym}, we can obtain the proper
but infeasible region, which satisfies
\begin{subequations}
\begin{align}
\label{Eq:Proper_Infeasible_Region1}
&M<\frac{pK+q}{p}d,~N<\frac{pK+q}{q}d\\
\label{Eq:Proper_Infeasible_Region2}
&M+N\geq (GK+1)d
\end{align}
\end{subequations}

From \eqref{Eq:Proper_Infeasible_Region1}, we have
$M+N<(pK+q)(1/p+1/q)d$. From \eqref{Eq:Proper_Infeasible_Region2},
we have $M+N\geq(GK+1)d$. Therefore, only if $(pK+q)(1/p+1/q)>
GK+1$, the proper but infeasible region will not be empty. It is not hard
to show that in the nonempty region, the values of $p,q$
need to satisfy the following quadratic inequality,
\begin{align}\label{Eq:Proper_Infeasible_Equation}
\Delta \triangleq K\left(\frac{p}{q}\right)^2-(G-1)K\frac{p}{q}+ 1>
0
\end{align}

$\Delta$ is a convex
function. Therefore, if \eqref{Eq:Proper_Infeasible_Equation} does
not hold when the value of $p/q$ achieves its minimum or maximum, it
will not hold for other values of $p$ and $q$. To find the cases that are proper but infeasible, we first check
whether \eqref{Eq:Proper_Infeasible_Equation} is satisfied when
$p/q$ achieves its minimum or maximum.

Since in \eqref{Eq:Necessary_Condition_Compress},
$\mathcal{I}_{\mathrm{A}},\mathcal{I}_{\mathrm{B}} \subseteq
\{1,\cdots,G\}$ and
$\mathcal{I}_{\mathrm{A}}\cap\mathcal{I}_{\mathrm{B}}=\varnothing$,
we have
$\mathcal{I}_{\mathrm{A}}\cup\mathcal{I}_{\mathrm{B}}\subseteq
\{1,\cdots,G\}$ and
$\mathcal{I}_{\mathrm{A}}\cap\mathcal{I}_{\mathrm{B}}=\varnothing$.
Therefore,
$|\mathcal{I}_{\mathrm{A}}|\leq G-1$,
$|\mathcal{I}_{\mathrm{B}}|\leq G-1$ and
$|\mathcal{I}_{\mathrm{A}}|+|\mathcal{I}_{\mathrm{B}}|\leq G$. From
the definition of $p$ and $q$ after \eqref{Eq:Constraint_ICI_free_Multiple}, we can derive that,
\begin{align}
\label{Eq:PQ_Region}
  \left\{\begin{array}{c}
           1\leq p \leq G-1,~1\leq q \leq (G-1)K\\
           Kp+q\leq GK
         \end{array}
  \right.
\end{align}
From \eqref{Eq:PQ_Region}, it is
easy to show that when $p=1,~q=(G-1)K$, $p/q=1/((G-1)K)$ achieves the
minimum, while when $p=(G-1),~q=1$, $p/q=(G-1)$ is the maximum.

When $p=1,~q=(G-1)K$, we have $\Delta=1/((G-1)^2K)$. Hence, \eqref{Eq:Proper_Infeasible_Equation} holds  for all $G,K$. Substituting the values of $p,~q$ into \eqref{Eq:Necessary_Condition_Compress_Sym1}, we have $\max\{M,(G-1)KN\}\geq GKd$, i.e., \eqref{Eq:Proper_Infeasible_Condition1}, which is one necessary condition that cannot be derived from the proper condition.

When $p=(G-1),~q=1$, we have $\Delta=1>0$.
Consequently, \eqref{Eq:Proper_Infeasible_Equation} still holds for all $G,K$. Substituting the values of $p,~q$ into \eqref{Eq:Necessary_Condition_Compress_Sym1}, we have $\max\{(G-1)M,N\}\geq((G-1)K+1)d$, i.e., \eqref{Eq:Proper_Infeasible_Condition2}, which is another necessary condition that cannot be derived from the proper condition.
\end{proof}

\section*{Acknowledgment}
The authors wish to thank Prof. Zhi-Quan (Tom) Luo for his constructive discussions. We also thank
the anonymous reviewers for providing a number of helpful comments.

\ifCLASSOPTIONcaptionsoff
  \newpage
\fi



%
\bibliographystyle{IEEEtran}


%

\begin{biography}[{\includegraphics[width=1in,height=1.25in,clip,keepaspectratio]{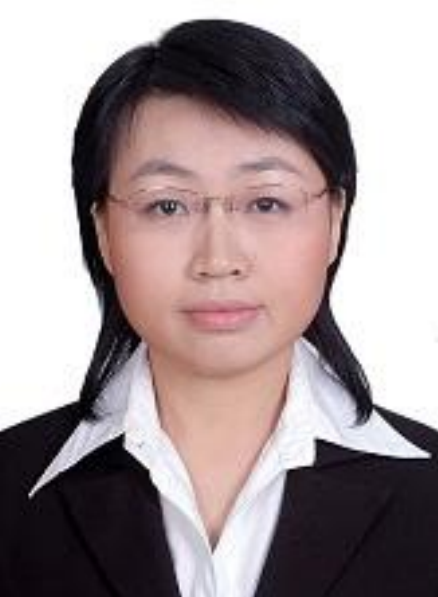}}]
{Tingting Liu} (S'09--M'11) received the B.S. and Ph.D. degrees in signal
and information processing from Beihang University, Beijing, China,
in 2004 and 2011, respectively.

From December 2008 to January 2010, she was a Visiting Student with
the School of Electronics and Computer Science, University of
Southampton, Southampton, U.K. She is currently a Postdoctoral Research
Fellow with the School of Electronics and Information Engineering,
Beihang University. Her research interests include wireless communications and signal processing, the degrees of freedom (DoF) analysis and interference alignment transceiver design in interference channels, energy efficient transmission strategy design and Joint transceiver design for multicarrier and multiple-input multiple-output communications. She received the awards of the \emph{2012 Excellent Doctoral Thesis in Beijing} and the \emph{2012 Excellent Doctoral Thesis in Beihang University}.
\end{biography}

\begin{biography}[{\includegraphics[width=1in,height=1.25in,clip,keepaspectratio]{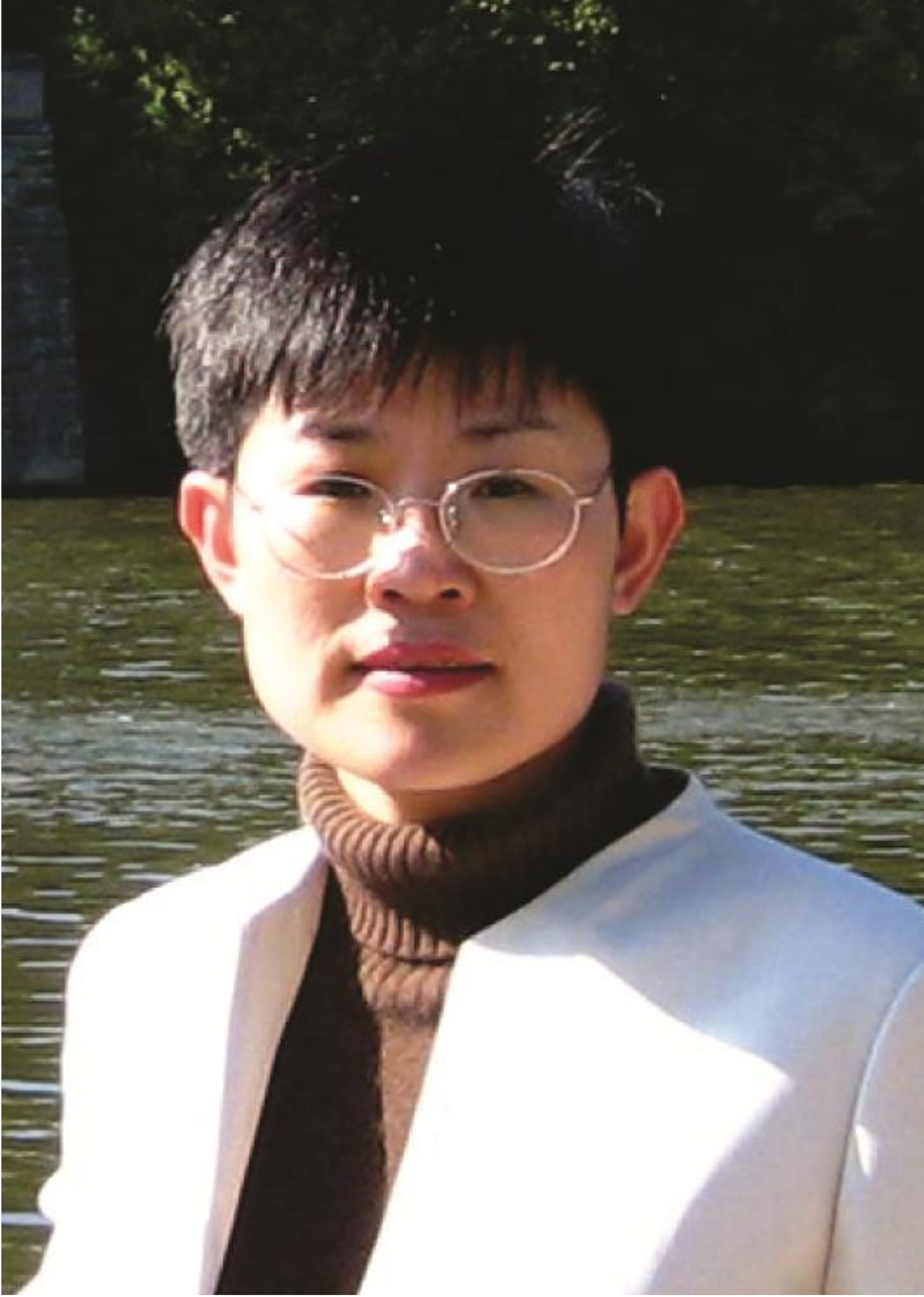}}]
{Chenyang Yang} (SM'08) received the M.S.E and Ph.D. degrees in
electrical engineering from Beihang University (formerly Beijing
University of Aeronautics and Astronautics), Beijing, China, in 1989
and 1997, respectively.

She is currently a Full Professor with the School of Electronics and
Information Engineering, Beihang University. She has published
various papers and filed many patents in the fields of signal
processing and wireless communications. Her recent research interests include network MIMO (CoMP, coordinated multi-point), energy efficient transmission (GR, green radio) and interference alignment.

Prof. Yang was the chair of Beijing chapter of IEEE Communications Society during 2008 to 2012. She has served as a Technical Program Committee Member for
many IEEE conferences, such as the IEEE International Conference on
Communications and the IEEE Global Telecommunications Conference.
She currently serves as an Associate Editor for \textsc{IEEE
Transactions on Wireless Communications},
the Membership Development Committee Chair of Asia Pacific Board of IEEE Communications Society,
an Associate
Editor-in-Chief of the \emph{Chinese Journal of Communications}, and
an Associate Editor-in-Chief of the \emph{Chinese Journal of Signal
Processing}. She was nominated as an Outstanding Young Professor of
Beijing in 1995 and was supported by the First Teaching and Research
Award Program for Outstanding Young Teachers of Higher Education
Institutions by Ministry of Education (P.R.C. ``TRAPOYT'') during
1999 to 2004.
\end{biography}




\end{document}